\documentclass[letterpaper, 11pt]{article}

\usepackage[english]{babel}
\usepackage[utf8x]{inputenc}
\usepackage[T1]{fontenc}

\usepackage[margin=0.87in, marginparsep=0pt,marginparwidth=0.75in]{geometry}

\usepackage{amsmath}
\usepackage{amsthm}
\usepackage{amsfonts}
\usepackage{graphicx}
\usepackage{subcaption}
\usepackage{tikz,pgfplots}
\usepackage{xspace, authblk}
\usetikzlibrary{patterns}
\usetikzlibrary{backgrounds}
\usetikzlibrary{scopes}
\usetikzlibrary{arrows, automata, positioning}
\usetikzlibrary{decorations.markings}
\usepackage[colorinlistoftodos]{todonotes}
\usepackage[colorlinks=true, allcolors=blue]{hyperref}
\usepackage{comment}
\usepackage{thm-restate}

\newtheorem{theorem}{Theorem}
\newtheorem{lemma}[theorem]{Lemma}
\newtheorem{proposition}[theorem]{Proposition}
\newtheorem{claim}[theorem]{Claim}
\newtheorem{definition}[theorem]{Definition}
\newtheorem{corollary}[theorem]{Corollary}
\newtheorem{observation}[theorem]{Observation}

\newcommand{\eps}{\varepsilon}
\newcommand{\cI}{\ensuremath{\mathcal{I}}}

\newcommand{\cT}{\ensuremath{\mathcal{T}}}
\newcommand{\cV}{\ensuremath{\mathcal{V}}}
\newcommand{\cB}{\ensuremath{\mathcal{B}}}
\newcommand{\cN}{\ensuremath{\mathcal{N}}}
\newcommand{\cM}{\ensuremath{\mathcal{M}}}
\newcommand{\cH}{\ensuremath{\mathcal{H}}}
\newcommand{\cL}{\ensuremath{\mathcal{L}}}

\newcommand{\abs}[1]{\ensuremath{|#1|}}
\newcommand{\floor}[1]{\ensuremath{\lfloor #1 \rfloor}}

\newcommand{\DSP}{\textsc{DSP}\xspace}
\newcommand{\GSP}{\textsc{GSP}\xspace}
\newcommand{\DSPS}{\textsc{Square-DSP}\xspace}
\newcommand{\GSPS}{\textsc{Square-GSP}\xspace}
\newcommand{\BP}{\textsc{Balanced Partition}}
\newcommand{\Part}{\textsc{Partition}}
\newcommand{\TwoDK}{\textsc{2D Knapsack}}

\ifdefined\DEBUG
\newcommand{\fab}[1]{\textcolor{red}{#1}}
\newcommand{\kam}[1]{\textcolor{cyan}{#1}}
\newcommand{\afr}[1]{\textcolor{violet}{#1}}
\newcommand{\wal}[1]{\textcolor{magenta}{#1}}

\def\rem#1{\marginpar{\raggedright\scriptsize #1}}
\newcommand{\fabr}[1]{\rem{\textcolor{red}{$\bullet$ F: #1}}}
\newcommand{\kamr}[1]{\rem{\textcolor{cyan}{$\bullet$ K: #1}}}
\newcommand{\afrr}[1]{\rem{\textcolor{violet}{$\bullet$ A: #1}}}
\newcommand{\walr}[1]{\rem{\textcolor{magenta}{$\bullet$ W: #1}}}
\else
\newcommand{\fab}[1]{#1}
\newcommand{\kam}[1]{#1}
\newcommand{\afr}[1]{#1}
\newcommand{\wal}[1]{#1}

\newcommand{\fabr}[1]{}
\newcommand{\kamr}[1]{}
\newcommand{\afrr}[1]{}
\newcommand{\walr}[1]{}
\fi

\title{Approximation Algorithms for\\Demand Strip Packing}
\author[1]{Waldo G\'alvez\thanks{Email: \texttt{galvez@in.tum.de} Supported by the European Research Council, Grant Agreement No. 691672, project APEG.}}
\author[2]{Fabrizio Grandoni\thanks{Email: \texttt{fabrizio@idsia.ch}. Partially supported by the SNF Excellence Grant 200020B\_182865.}}
\author[2]{Afrouz Jabal Ameli\thanks{Email: \texttt{afrouz@idsia.ch}}}
\author[3]{Kamyar Khodamoradi\thanks{Email: \texttt{kamyar.khodamoradi@uni-wuerzburg.de} Partially supported by Deutsche Forschungsgemeinschaft (DFG, German Research Foundation) - Project number 399223600. This project was carried out in part when the author was a postdoctoral researcher at IDSIA, USI-SUPSI, Switzerland.}}
\affil[1]{Technical University of Munich, Germany}
\affil[2]{IDSIA, USI-SUPSI, Switzerland}
\affil[3]{University of W\"urzburg, Germany}
\date{}

\begin{document}
\maketitle
\thispagestyle{empty}

\abstract{
\noindent In the Demand Strip Packing problem (DSP)\kam{,} we are given a time interval and a collection of tasks, each\kamr{deleted ``one'' here} characterized by a processing time and a demand for a given resource (such as electricity, computational power, etc.). A  feasible solution consists of a schedule of the tasks within the mentioned time interval. Our goal is to minimize the peak resource consumption, i.e. the maximum total demand of tasks executed at any point in time. 

It is known that DSP is NP-hard to approximate below a factor $3/2$, and standard techniques for related problems imply a (polynomial-time) $2$-approximation. Our main result is a $(5/3+\eps)$-approximation algorithm for any constant $\eps>0$. We also achieve best-possible approximation factors for some relevant special cases. 
}

\newpage
\setcounter{page}{1}
\section{Introduction}\label{sec:intro}

Consider the following scenario: we are given a time interval and a collection of tasks, where each task is characterized by a processing time (no longer than the time interval) and a demand for a given resource. A feasible solution consists of a schedule of all the tasks within the mentioned time interval, and our goal is to minimize the peak resource consumption, i.e. the maximum total demand of tasks scheduled at any point in time. It is easy to imagine concrete applications of this scenario; for example, the considered resource might be electricity, bandwidth along a communication channel, or computational power. 

The above scenario can be naturally formalized via the following \textsc{Demand Strip Packing} problem (DSP). We interpret the time interval as a path graph $G=(V,E)$ with $W$ edges, where each edge is interpreted as a time slot where we can start to process a task. Let $\cI=\{1, \dots, n\}$ be the set of tasks, where task $i$ has integer \emph{processing time} (or \emph{width}) $w(i)\in [\fab{1},W]$ and integer \emph{demand} (or \emph{height}) $h(i)\geq 0$. A feasible solution (or \emph{schedule} of the tasks) consists of a subpath $P(i)$ of $G$ for each $i\in \cI$ containing precisely $w(i)$ edges. 
Our goal is to minimize the peak resource consumption (or simply \emph{peak}) which is defined as
$$
\max_{e\in E}\sum_{i\in \cI: e\in P(i)}h(i).
$$ 
A problem closely related to DSP is the \textsc{Geometric Strip Packing} problem (GSP)\footnote{GSP is usually simply called \textsc{Strip Packing} in the literature. We added the word ``geometric'' to better highlight the differences between the two problems.}, which can be interpreted as a variant of DSP with an extra geometric packing constraint. Here we are given an axis-aligned half-strip of integer width $W$ (and unbounded height) and a collection of open rectangles (or tasks), where each rectangle $i$ has integer width $w(i)\in [\fab{1},W]$ and integer height $h(i)\geq 0$. Our goal is to find an axis-aligned non-overlapping packing of all the rectangles within the strip that minimizes the peak height, i.e. the maximum height spanned by any rectangle. Notice that one can reinterpret DSP as a variant of GSP, where the processing time and demand of each task correspond to the width and height of a rectangle, resp. (this also motivated our notation). A critical difference \fab{w.r.t.} GSP however is that DSP does not require to pack such rectangles geometrically\footnote{Or, equivalently, we can \fab{split such rectangles into unit-width vertical slices, and then pack such slices geometrically so that slices of the same rectangle appear consecutively in a horizontal sense.}}. 

Obviously, a feasible solution to GSP induces a feasible solution to DSP of no larger peak. The converse is however not true (see Figure~\ref{fig:gap_examples}), and consequently it makes sense to design algorithms specifically for DSP. We remark that there are applications that are better formalized by GSP than by DSP. In particular, this happens when each task requires a contiguous and fixed portion of the considered resource. For example, we might need to allocate consecutive frequencies or memory locations to each task: changing this allocation over time might be problematic. Another natural application of GSP is cutting rectangular pieces from a roll of some raw material (e.g., paper, metal, or leather). However, for other applications\kam{,} the geometric constraint in GSP does not seem to be necessary, and hence it makes sense to drop it (i.e., to rather consider DSP): this might lead to better solutions, possibly via simpler and/or more efficient algorithms. Consider for example the  
minimization of the peak energy consumption in smart-grids~\cite{KSKL13,YM17,RKS15}.

A straightforward reduction to the NP-complete \textsc{Partition} problem (similar to the one known for  GSP, see also \cite{THLW13}) shows that DSP is NP-hard to approximate below a factor $3/2$. Constant approximation algorithms for DSP are given in \cite{THLW13,YMML14}. However, \fab{a better} $2$-approximation can be obtained by applying an algorithm by Steinberg~\cite{Steinberg97} which was developed for GSP: the reason is that Steinberg uses area-based lower bounds that extend directly from GSP to DSP.

\subsection{Our Results and Techniques}

Our main result is as follows\footnote{The same result as in Theorem~\ref{thr:53} was achieved independently in \cite{DJKRT21}; their approach is however substantially different from ours.}. 
\begin{restatable}{theorem}{thmmainapx}\label{thr:53}
For any constant $\eps>0$, there is a polynomial-time deterministic $(5/3+\eps)$-approximation algorithm for \DSP. 
\end{restatable}

 The above approximation ratio matches the best\kam{-}known result for \GSP from Harren et al.~\cite{HJPS14}, achieved using dynamic programming based techniques to place almost all the rectangles except for a set of very small total area, followed by a careful and quite involved case distinction to pack these remaining rectangles. However, we remark that our algorithm is entirely different, and in particular it does not compute a geometric packing of tasks/rectangles. Furthermore, our analysis is substantially simpler. Notice that the result in \cite{HJPS14} does \emph{not} imply a $(5/3+\eps)$-approximation for DSP since some lower bounds used in their proofs do not hold necessarily for DSP.

\begin{figure}
    \centering
    \begin{subfigure}{0.39\textwidth}
    \centering
		\resizebox{0.85\textwidth}{!}{\begin{tikzpicture}
		\draw[white] (0,6.3) -- (1,6.3);
		\draw[thick] (0,0) -- (7,0);
		\draw[color=gray!70] (0,0) -- (0,5);
		\draw[color=gray!70] (7,0) -- (7,5);
		\draw[color=gray!70] (0,0) -- (-0.25,0);
		\draw (-0.25,0) node[anchor=east] {\color{gray!70}$0$};
		\draw[color=gray!70] (0,1) -- (-0.25,1);
		\draw (-0.25,1) node[anchor=east] {\color{gray!70}$1$};
		\draw[color=gray!70] (0,2) -- (-0.25,2);
		\draw (-0.25,2) node[anchor=east] {\color{gray!70}$2$};
		\draw[color=gray!70] (0,3) -- (-0.25,3);
		\draw (-0.25,3) node[anchor=east] {\color{gray!70}$3$};
		\draw[color=gray!70] (0,4) -- (-0.25,4);
		\draw (-0.25,4) node[anchor=east] {\color{gray!70}$4$};
		\draw[color=gray!70] (0,5) -- (-0.25,5);
		\draw (-0.25,5) node[anchor=east] {\color{gray!70}$5$};
		\draw (0,0) -- (0,-0.25);
		\draw (0,-0.25) node[anchor=north] {$0$};
		\draw (1,0) -- (1,-0.25);
		\draw (1,-0.25) node[anchor=north] {$1$};
		\draw (2,0) -- (2,-0.25);
		\draw (2,-0.25) node[anchor=north] {$2$};
		\draw (3,0) -- (3,-0.25);
		\draw (3,-0.25) node[anchor=north] {$3$};
		\draw (4,0) -- (4,-0.25);
		\draw (4,-0.25) node[anchor=north] {$4$};
		\draw (5,0) -- (5,-0.25);
		\draw (5,-0.25) node[anchor=north] {$5$};
		\draw (6,0) -- (6,-0.25);
		\draw (6,-0.25) node[anchor=north] {$6$};
		\draw (7,0) -- (7,-0.25);
		\draw (7,-0.25) node[anchor=north] {$7$};
		\draw (0,0) rectangle (4,1);
		\draw (2,0.5) node {$3$};
		\draw (0,1) rectangle (2,4);
		\draw (1,2.5) node {$1$};
		\draw (5,0) rectangle (7,3);
		\draw (6,1.5) node {$2$};
		\draw (3,3) rectangle (7,4);
		\draw (5,3.5) node {$4$};
		\draw (3,1) rectangle (4,2);
		\draw (3.5,1.5) node {$6$};
		\draw (4,0) rectangle (5,2);
		\draw (4.5,1) node {$7$};
		\draw (2, 1) rectangle (3, 3);
		\draw (2.5, 2) node {$8$};
		\draw[line width=0.8mm] (2, 3) rectangle (3, 4);
		\draw (2.5, 3.5) node {$5$};
		\draw[line width=0.8mm] (3, 2) rectangle (5, 3);
		\draw (4, 2.5) node {$5$};
		\end{tikzpicture}}
		\caption{DSP solution of peak $4$ whose corresponding optimal GSP solution has peak $5$.}\label{fig:gap_instance}
	\end{subfigure}
	\qquad
	\begin{subfigure}{0.48\textwidth}
	\centering
    \resizebox{0.73\textwidth}{!}{\begin{tikzpicture}[scale=0.65]
    \draw[thick] (0,0) -- (13,0);
    \draw[color=gray!70] (0,0) -- (0,11.5);
    \draw[color=gray!70] (13,0) -- (13,11.5);
    \draw[color=gray!70] (0,0) -- (-0.25,0);
    \draw (-0.25,0) node[anchor=east] {\color{gray!70}$0$};
    \draw[color=gray!70] (0,1) -- (-0.25,1);
    \draw (-0.25,1) node[anchor=east] {\color{gray!70}$1$};
    \draw[color=gray!70] (0,2) -- (-0.25,2);
    \draw (-0.25,2) node[anchor=east] {\color{gray!70}$2$};
    \draw[color=gray!70] (0,3) -- (-0.25,3);
    \draw (-0.25,3) node[anchor=east] {\color{gray!70}$3$};
    \draw[color=gray!70] (0,4) -- (-0.25,4);
    \draw (-0.25,4) node[anchor=east] {\color{gray!70}$4$};
    \draw[color=gray!70] (0,5) -- (-0.25,5);
    \draw (-0.25,5) node[anchor=east] {\color{gray!70}$5$};
    \draw[color=gray!70] (0,6) -- (-0.25,6);
    \draw (-0.25,6) node[anchor=east] {\color{gray!70}$6$};
    \draw[color=gray!70] (0,7) -- (-0.25,7);
    \draw (-0.25,7) node[anchor=east] {\color{gray!70}$7$};
    \draw[color=gray!70] (0,8) -- (-0.25,8);
    \draw (-0.25,8) node[anchor=east] {\color{gray!70}$8$};
    \draw[color=gray!70] (0,9) -- (-0.25,9);
    \draw (-0.25,9) node[anchor=east] {\color{gray!70}$9$};
    \draw[color=gray!70] (0,10) -- (-0.25,10);
    \draw (-0.25,10) node[anchor=east] {\color{gray!70}$10$};
    \draw[color=gray!70] (0,11) -- (-0.25,11);
    \draw (-0.25,11) node[anchor=east] {\color{gray!70}$11$};

    \draw (0,0) -- (0,-0.25);
    \draw (0,-0.25) node[anchor=north] {$0$};
    \draw (1,0) -- (1,-0.25);
    \draw (1,-0.25) node[anchor=north] {$1$};
    \draw (2,0) -- (2,-0.25);
    \draw (2,-0.25) node[anchor=north] {$2$};
    \draw (3,0) -- (3,-0.25);
    \draw (3,-0.25) node[anchor=north] {$3$};
    \draw (4,0) -- (4,-0.25);
    \draw (4,-0.25) node[anchor=north] {$4$};
    \draw (5,0) -- (5,-0.25);
    \draw (5,-0.25) node[anchor=north] {$5$};
    \draw (6,0) -- (6,-0.25);
    \draw (6,-0.25) node[anchor=north] {$6$};
    \draw (7,0) -- (7,-0.25);
    \draw (7,-0.25) node[anchor=north] {$7$};
    \draw (8,0) -- (8,-0.25);
    \draw (8,-0.25) node[anchor=north] {$8$};
    
    \draw (9,0) -- (9,-0.25);
    \draw (9,-0.25) node[anchor=north] {$9$};
    
    \draw (10,0) -- (10,-0.25);
    \draw (10,-0.25) node[anchor=north] {$10$};
    
    \draw (11,0) -- (11,-0.25);
    \draw (11,-0.25) node[anchor=north] {$11$};
    
    \draw (12,0) -- (12,-0.25);
    \draw (12,-0.25) node[anchor=north] {$12$};
    
    \draw (13,0) -- (13,-0.25);
    \draw (13,-0.25) node[anchor=north] {$13$};
    \draw (0,5) rectangle (6,11);
    \draw (3,8) node {$1$};
    \draw (7,5) rectangle (13,11);
    \draw (10,8) node {$2$};
    \draw (0,0) rectangle (5,5);
    \draw (2.5,2.5) node {$3$};
    \draw (8,0) rectangle (13,5);
    \draw (10.5,2.5) node {$4$};
    \draw (5,0) rectangle (8,3);
    \draw (6.5,1.5) node {$5$};
    \draw (5,3) rectangle (7,5);
    \draw (6,4) node {$6$};
    \draw (6,7) rectangle (7,8);
    \draw (6.5,7.5) node {$8$};
    \draw (6,8) rectangle (7,9);
    \draw (6.5,8.5) node {$9$};
    \draw (6,9) rectangle (7,10);
    \draw (6.5,9.5) node {$10$};
    \draw (6,10) rectangle (7,11);
    \draw (6.5,10.5) node {$11$};
    \draw[line width=0.8mm] (6, 5) rectangle (7, 7);
    \draw (6.5,6) node {$7$};
    \draw[line width=0.8mm] (7, 3) rectangle (8, 5);
    \draw (7.5,4) node {$7$};
    
    \end{tikzpicture}}
    \caption{\textsc{Square-DSP} solution of peak $11$ whose corresponding optimal GSP solution has peak at least $12$.}
    \label{fig:gap_instance_PEC_squares}
    \end{subfigure}
    \caption{Gap instances between DSP and GSP.}\label{fig:gap_examples}
\end{figure}
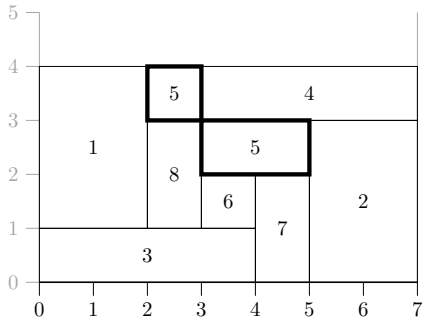
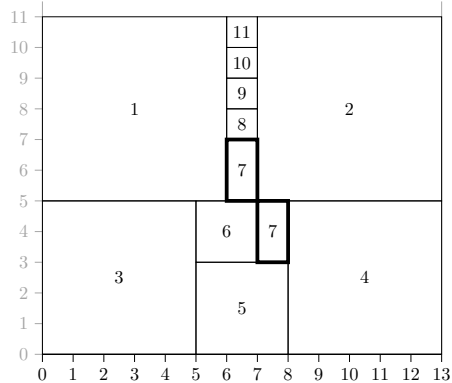


We also achieve improved approximation algorithms for relevant special cases of \DSP. We obtain a PTAS for the special case where the demand of each task is much lower than the optimal peak $OPT$. This captures applications where each job consumes a relatively small amount of the available resource (think about the electricity consumption of large-scale systems such as cities or countries).

\begin{restatable}{theorem}{lightPTAS}\label{lem:light-PTAS} Given $\varepsilon>0$ small enough, there exists $\delta>0$ and a polynomial time algorithm such that, given an instance of \DSP with optimal value $OPT$ and consisting solely of tasks having height at most $\delta\cdot OPT$, it computes a $(1+O(\varepsilon))$-approximate solution. \end{restatable}

Motivated by the special case of GSP and related problems where all rectangles are squares, we also study the special case of \DSP where $h(j)=w(j)$ for all tasks (the \DSPS problem).
The $3/2-\eps$ hardness of approximation extends to this case (see Appendix~\ref{sec:hardness}), and we are still able to show that there is a gap between \DSP and \GSP (see Figure~\ref{fig:gap_instance_PEC_squares} and Appendix~\ref{sec:DSPvsGSP}). However, in this case\kam{,} we are able to provide an optimal $3/2$-approximation. 
\begin{restatable}{theorem}{squareDSP}\label{thm:squareDSP} There is a deterministic polynomial-time $3/2$-approximation for \DSPS. 
\end{restatable}

At \kam{a} high level, our approach is based on a classification of tasks into groups depending on their heights and widths. We carefully schedule some groups first, so that their \emph{demand profile} has a convenient structure. Here\kam{,} by demand profile we simply mean the total demand of the already scheduled tasks over each edge. The structure of the demand profile allows us to pack the remaining groups (intuitively, on top of such profile) in a convenient way. We critically exploit the fact that, differently from GSP, we only care about the total demand on each edge. 
This allows us to adapt techniques from Bin Packing or Makespan Minimization (see Lemmas~\ref{lem:packing-narrow-sorted} and \ref{lem:packing-narrow-Qtsorted}).

\subsection{Related Work}

GSP generalizes famous problems such as Makespan Minimization on identical machines~\cite{CB76} (here all the rectangles have width $1$ and $W$ corresponds to the number of processors) or Bin Packing~\cite{CCGMV13} (here all the rectangles have height $1$ and the height of the solution corresponds to the number of bins). Consequently, it is known that for any $\varepsilon>0$, there is no $(3/2-\varepsilon)$-approximation for the problem unless P$=$NP. The first non-trivial approximation algorithm for GSP, with \kam{an} approximation ratio \kam{of} $3$, was given by Baker, Coffman, and Rivest \cite{BCR80}. After a series of very technical and involved refinements~\cite{CGJT80,S80,S94,Steinberg97,HS09}, the current best approximation factor for the problem is $(5/3+\varepsilon)$ due to Harren et al.~\cite{HJPS14}. GSP has been also studied in the pseudopolynomial setting, i.e., when $W=n^{O(1)}$~\cite{JT10,NW16, AKPP17, GGIK16, HJRS20, JR19,JR19ESA} and in the asymptotic setting, i.e. when the optimal value is assumed to be large~\cite{KR00, JS09}. In both cases\kam{,} approximation algorithms and almost matching lower bounds have been developed.

There is a very rich line of research on generalizations and variants of \DSP such as online versions~\cite{LLW19,LLW20}, tasks with availability constraints or time windows~\cite{YMML14, YM17, KSKL13}, \kam{a} mixture of preemptable and non-preemptable tasks~\cite{RKS15} or generalized cost functions based on the demand at each edge~\cite{BHLWY16,LLW20}. The variant of \DSP with the extra feature of interrupting the tasks is known as \textsc{Strip Packing with Slicing}, for which there exists an FPTAS~\cite{ABCGJKLP13}; on the other hand, the case of \DSP is still hard to approximate by a factor better than $3/2$ as noted by Tang et al.~\cite{THLW13}.



Another problem closely related to \DSP is \textsc{Parallel Job Scheduling}. Here we are given a set of jobs and $m$ machines, where each job is characterized by a processing time and a number of machines where the job must be processed simultaneously (these machines do not need to be contiguous), and the goal is to minimize the makespan. The same hardness of approximation applies in this case, but interestingly an almost tight $(3/2+\eps)$-approximation algorithm has been developed~\cite{J12} and also a pseudopolynomial $(1+\varepsilon)$-approximation is known~\cite{JT10}. See~\cite{DMT04} for a comprehensive survey on the problem and its many variants.

It is also worth mentioning another case where the distinction between geometric and demand-based packing plays a substantial role: the \textsc{Unsplittable Flow on a Path} problem (UFP) \cite{AGLW18,GMWZ18,BGKMW15,GIU15} and the \textsc{Storage Allocation} problem (SAP) \cite{MW15,MW20}. In both problems\kam{,} we are given a path graph with edge capacities, and tasks specified by a subpath, a demand (or height), and a profit. In both problems\kam{,} the goal is to select a maximum profit subset of tasks that can be \emph{packed} while respecting edge capacities. For UFP,  analogously to DSP, we require that the total demand of the selected tasks on each edge $e$ is at most the capacity of $e$. For SAP, analogously to GSP, we interpret each task as a rectangle (with \kam{the} width given by its number of edges) and, intuitively, we need to pack such rectangles non-overlappingly below the capacity profile. Notice that, differently from DSP and GSP, here the path associated with each task is fixed in the input. Furthermore, not all the tasks need to be packed.

Finally, in the \textsc{Dynamic Storage Allocation}  problem (DSA) the setting is analogous to SAP but, similarly to GSP, we are asked for an embedding of all the rectangles minimizing the peak height, i.e. the maximum height reached by any rectangle (in particular, there are no edge capacities). Notice that in DSA a lower bound is provided by the peak demand, i.e. the maximum over the edges $e$ of the sum of the heights of rectangles whose path uses $e$. Buchsbaum et al.~\cite{BKKRT04} studied in detail the relation between the optimal peak height and the peak demand, providing examples where these values differ by a constant factor. The authors also present a $(2+\varepsilon)$-approximation for DSA that provides guarantees even when compared with the peak demand. 

%


\subsection{Organization}

\walr{Please check the reordering}
We start by introducing in Section~\ref{sec:prelim} some useful definitions and known results. As a warm-up, in Section \ref{sec:2apx} we present a very simple $2$-approximation for DSP that allows us to illustrate part of our ideas. Then in Section~\ref{sec:approx-alg} we present our main result, namely a $(5/3+\varepsilon)$-approximation for DSP. \wal{Due to space constraints, details about the gap instances in Figure~\ref{fig:gap_examples} and the results for special cases of DSP (Theorems~\ref{lem:light-PTAS} and \ref{thm:squareDSP}) are deferred to the Appendix}.


%
%


\section{Preliminaries}\label{sec:prelim}

Let $e_1,\ldots,e_{W}$ be the edges of $G$ from left to right. Recall that a feasible solution or schedule $P(\cdot)$ specifies a subpath $P(i)$ of $G$ of length $w(i)$ for each task $i$. Sometimes it is convenient to consider a \emph{partial} schedule $P(\cdot)$ which specified the path of a subset $\cI'$ of tasks only (it is convenient to consider $P(i)$ as an \emph{empty path} for the remaining tasks). 
We call this a schedule of $\cI'$.

Let us define, for a given subset $\cI'$ of tasks, $h_{\max}(\cI'):= \max_{i\in \cI'}{h(i)}$ and $h(\cI'):= \sum_{i\in\cI'}{h(i)}$. We define analogously $w_{max}(\cI')$ and $w(\cI')$ w.r.t. widths. Let also $a(\cI'):=\sum_{i\in\cI'}{a(i)}$, where $a(i):= h(i) \cdot w(i)$ corresponds to the \emph{area} of task $i$. We will start by showing a couple of simple lower bounds for the optimal peak $OPT$ that will be used extensively along this work.


\begin{proposition}\label{prop:LB_OPT} $OPT\ge \max \{h_{max}(\cI), \sum_{i\in \cI: w(i)>W/2}{h(i)}, a(\cI)/W\}$.
\end{proposition}

\begin{proof} Since the total demand of any edge used by the task of largest height is at least $h_{\max}(\cI)$, it holds that $OPT \ge h_{\max}(\cI)$. Also notice that, in any scheduling, the tasks of width larger than $W/2$ use the edge $e_{\lceil W/2 \rceil}$, being then the total demand of this edge (and consequently $OPT$) at least $\sum_{i\in \cI: w(i)>W/2}{h(i)}$. Finally, the last bound follows from an averaging argument and the fact that the sum over the edges of the total demand on each edge is equal to $a(\cI)$. \end{proof}




\subsection{Demand Profile and Left-Pushing}\label{sec:left-push}

Consider 
a schedule $P(\cdot)$ of $\cI' \subseteq \cI$. We define the \emph{demand profile} $h(P)$ of $P(\cdot)$ as the vector that stores for each edge $e$ the total demand $\sum_{i\in \cI':e\in P(i)}h(i)$ of the tasks whose path contains $e$ (if the path of $i$ is not specified, then $i$ does not contribute to the demand profile). Since $W$ can be exponential in $n$, we need to store the demand profile in a more efficient way. This can be done by noticing that the number of times the total demand changes from an edge to the next one is at most $2n$ (when a task starts or finishes). Hence we just need to store the edges where the demand profile changes value and the corresponding demand. In particular, we can efficiently store the demand profile. Furthermore, we can efficiently update it, e.g., when augmenting an existing schedule by specifying the path $P(i)$ of one more task $i$, or when we modify the value of some $P(i)$ by \emph{shifting} tasks as we will discuss later.

%

Given a schedule $P(\cdot)$ of $\cI'\subseteq \cI$ and $i\in \cI'$, a \emph{left-shifting} of $i$ in $P(\cdot)$ is the operation of replacing $P(i)$ with the path $P'(i)$ of length $w(i)$ that starts one edge to the left of $P(i)$. 
Clearly\kam{,} this operation is allowed only if $P(i)$ does not start at the leftmost edge of $G$. Consider a schedule $P(\cdot)$ with peak $\pi$, and let $\pi'\geq \pi$. A \emph{$\pi'$-left-pushing} of $P(\cdot)$ is the operation of iteratively performing left-shiftings in any order until it is not possible to continue\kamr{omitted ``,''} while guaranteeing that the peak is always at most $\pi'$. We will critically use left-pushings in our algorithms. Notice that a left-pushing can be computed in polynomial time (see Appendix~\ref{sec:light-PTAS} for \fab{some more} details). 

Intuitively, left-pushing accumulates the demand over the first edges while inducing a non-increasing demand profile to the right. For a node $t^*$ of the path and a value $Q\geq 0$, we will say that a (possibly partial) schedule $P(\cdot)$ is $(Q,t^*)$-sorted if the corresponding demand on the edges to the left of $t^*$ is at least $Q$\kamr{omitted ``,''} and on the edges to the right of $t^*$ the demand 
profile is non-increasing (see Figure~\ref{fig:classif}); if $t^*$ is the leftmost node we just say that the schedule is sorted. Our algorithms will first schedule some tasks\kamr{omitted ``,''} and then perform a left-pushing. After that, it will be possible to schedule the remaining tasks in a convenient way thanks to the properties of the resulting demand profile.

\subsection{Container-based Scheduling}

Similar\kamr{changed ``Similarly'' to ``Similar''} to recent work on related rectangle packing problems (e.g., \cite{GGHIKW17,BCJPS09}), we will exploit a \emph{container-based} scheduling approach. 
A \emph{container} $C$ can be interpreted as an artificial task, with its own width $w(C)$ (i.e. a number of edges) and height $h(C)$. Furthermore\kam{,} it is classified as \emph{vertical} or \emph{horizontal}, with a meaning which is explained later. The containers are scheduled as usual tasks in a DSP instance (in particular by defining a path $P(C)$ for each container $C$), with the goal of minimizing the peak $\pi$. We also define a packing of tasks into containers $C$ respecting the following constraints: if $C$ is vertical, the tasks $\cI(C)$ packed into $C$ must have height at most $h(C)$ and total width at most $w(C)$; if $C$ is horizontal, tasks $\cI(C)$ must have width at most $w(C)$ and total height at most $h(C)$. Intuitively, the tasks packed into a vertical (resp., horizontal) container induce a geometric packing of the rectangles associated with each task into the rectangle corresponding to the container, where the task rectangles are packed non-overlappingly one next to the other (resp., one on top of the other).
Any such packing and schedul\kam{ing} of containers naturally induces a schedule of the tasks: if $C$ is horizontal, tasks $\cI(C)$ are all scheduled starting on the leftmost edge of $P(C)$. Otherwise, tasks $\cI(C)$ are scheduled one after the other starting at the leftmost edge of $P(C)$. It is hopefully clear to the reader that the demand profile of such \kam{a} schedule of the tasks is dominated by the demand profile of the containers' schedule. In particular, if the latter has peak $\pi$, then the corresponding schedule of the tasks has \kam{a} no larger peak.

The general strategy is then as follows: we first show that there exists a convenient packing of tasks into a constant number of containers\kamr{removed ``,''} and that there exists a scheduling $P^*(\cdot)$ of these containers with \kam{a} small peak $\pi$. We also require that these containers are \emph{guessable}, meaning that we can guess their sizes by exploring a polynomial number of options. 
Once we guessed the correct set of containers, a $\pi$-left-pushing of $P^*(\cdot)$ can be computed by brute force (since they are constantly many tasks). 
Finally\kam{,} we pack tasks into containers, inducing a schedule of the tasks with peak $\pi$. 

This final step can be performed (almost completely) via a reduction to the 
\textsc{Generalized Assignment} problem (GAP). Recall that in GAP we are given a set of $k$ bins, where each bin $j$ has an associated capacity $C_j\geq 0$, and a set of $n$ items. For each item $i$ and bin $j$, the input specifies a size $s_{ij}\geq 0$ and a profit $p_{ij}\geq 0$ of item $i$ w.r.t. bin $j$. A feasible solution assigns each item to some bin\kamr{removed ``,''} so that the total size of the items assigned to each bin $j$ is at most $C_j$. Our goal is to maximize the total profit associated with this assignment. GAP admits a PTAS in the case of a constant number of bins (see Section E.2 in \cite{GGHIKW17arxiv}). 

\begin{lemma}\label{lem:containersPackPTAS} For any constant $\varepsilon'>0$, given a set of tasks $\cI'$ that can be packed into a given set of containers of constant cardinality, there is a polynomial\kam{-}time algorithm to pack $\cI''\subseteq \cI'$ with $a(\cI'')\geq (1-\varepsilon')a(\cI')$ into the mentioned containers. 
\end{lemma}

\begin{proof} We define a GAP instance as follows: we create one bin per container, where the capacity of the bin is equal to the width of the container if it is vertical or the height of the container if it is horizontal. For each task $i$ we define an item that \kamr{changed ``which'' to ``that''} has uniform profit equal to its area $a(i)$ over all the bins. Given a task $i$ and a vertical  (resp., horizontal) container $j$, the size $s_{ij}$ of $i$ into bin $j$ is set to $w(i)$ (resp., $h(i)$) if task $i$ can be packed into container $j$ according to the mentioned rules. Otherwise we set $s_{ij}=+\infty$. The claim follows by applying the aforementioned PTAS for GAP with parameter $\varepsilon'$. 
\end{proof}

Notice that the above lemma allows us to pack all the tasks but a subset of small total area, hence we need to schedule somehow such \emph{leftover} tasks. This is not necessarily a trivial task; indeed, such tasks, though of small area, might have large height, and hence scheduling them on top of the rest might substantially increase the peak. To circumvent this issue we will identify special containers reserved \kam{for}\kamr{changed ``to'' to ``for''} tasks of large height where we will be able to pack all such tasks with no leftovers. We will then apply the PTAS from Lemma \ref{lem:containersPackPTAS} only to the remaining tasks and containers.

\section{A Simple $2$-Approximation for DSP}
\label{sec:2apx}

In order to introduce part of our ideas, in this section\kam{,} we present a simple $2$-approximation for DSP. As mentioned before, a $2$-approximation can also be achieved via Steinberg's algorithm~\cite{Steinberg97}, however\kam{,} that algorithm is substantially more complex (not surprisingly since it computes a \emph{geometric packing} of tasks interpreted as rectangles like in GSP). 

The following lemma exploits a modification of Next-Fit-Decreasing~\cite{CCGMV13}, the well\kam{-}known approximation algorithm for Bin Packing.

%
%
%

\begin{lemma}\label{lem:packing-narrow-sorted} Let $P(\cdot)$ be a sorted schedule of $\cI'\subseteq  \cI$  with peak at most $\pi$. For $\cI'':=\cI\setminus \cI'$, assume that:\fabr{Looks good to me}
\begin{itemize} 
\item $\pi\ge h_{\max}(\cI'') + \max\{a(\cI)/W, h_{\max}(\cI'')\}$,
\item $w_{\max}(\cI'')\le W/2$, and
\item $(W-w_{\max}(\cI''))(\pi-h_{\max}(\cI'')) + w_{\max}(\cI'')\cdot h_{\max}(\cI'') \ge a(\cI)$. \end{itemize}
Then it is possible to compute in polynomial time a schedule of $\cI$ having peak at most $\pi$. \end{lemma}

\begin{proof} 
By slightly abusing notation, we will next use an edge label $e$ also to denote the position $i$ of $e$ in the sequence $e_1,\ldots,e_{W}$ of edges from left to right.
We do not modify the schedule of $\cI'$\kamr{removed ``,''} and schedule the remaining tasks $\cI''$ as follows. Let us fix an arbitrary order for tasks in $\cI''$, and let us initially define $e^{check}$ to be the leftmost edge. We scan completely the list of tasks $\cI''$ and, if the current task $i$ can be scheduled starting on edge $e^{check}$ while maintaining a peak of at most $\pi$, we do that and remove $i$ from $\cI''$; otherwise\kam{,} we keep $i$ in $\cI''$ and try with the next task. Once we consider the final task, we update $e^{check}$ to be the leftmost edge to the right of the current $e^{check}$ whose demand is different from the demand of the current $e^{check}$. We iterate the procedure on the new $e^{check}$ until all tasks are scheduled or we identify a task $i$ which cannot be scheduled.

It is not difficult to see that with this procedure, the demand profile from $e^{check}$ to its right is always non-increasing (restricted to these edges, scheduling a task is equivalent to summing up two non-decreasing profiles), and none of the remaining tasks can fit in any one of the edges to the left of $e^{check}$ (as we actually tried to place them there but it was not possible). Notice also that, if this procedure manages to schedule all the tasks, then the claimed peak is automatically achieved. So we will assume by contradiction that this is not the case. 

Let $i$ be a task that could not be scheduled. This could only happen due to $i$ being too wide for the current edge $e^{check}$ where it should be scheduled (and hence for any subsequent edge). This implies that $e^{check}>W-w(i)$ and hence there are more than $W-w(i)$ edges having demand larger than $\pi - h(i)$. \fab{Thus} the total area  
of the scheduled tasks plus task $i$ is strictly larger than 
$$
A(h(i),w(i)) := (W-w(i))\cdot (\pi - h(i)) + h(i)\cdot w(i).
$$ 
This expression is decreasing both as a function of $h(i)$ and as a function of $w(i)$. Indeed,
$$
\frac{\partial}{\partial h(i)} A(h(i),w(i)) = 2w(i)-W \le 0, \text{ and }  \frac{\partial}{\partial w(i)} A(h(i),w(i)) = 2h(i)-\pi \le 0,
$$ 
where we used the fact that, by assumption, $w(i)\le \frac{W}{2}$ and $\pi\ge 2h_{\max}(\cI'')\geq 2h(i)$. We conclude that
$$
A(h(i),w(i)) \ge A(h_{\max}(\cI''),w_{\max}(\cI'')) = (W-w_{\max}(\cI''))(\pi-h_{\max}(\cI'')) + w_{\max}(\cI'')\cdot h_{\max}(\cI'') \ge a(\cI),
$$ 
where in the last inequality we used the third assumption. This is a contradiction since a subset of tasks would have area strictly larger than the total area $a(\cI)$.
\end{proof}

We are now ready to provide a simple $2$-approximation.
\begin{corollary}\label{cor:2-apx} 
There exists a deterministic $2$-approximation for DSP. 
\end{corollary}

\begin{proof} Let $\cI$ be an instance of DSP. We will first schedule the tasks $\cI'$ having width larger than $W/2$ starting on the leftmost edge. Let $\cI'':=\cI\setminus \cI'$. This partial schedule is sorted and has peak $\sum_{i\in \cI'}{h(i)} \le M:=\max \{h_{max}(\cI), \sum_{i\in \cI'}{h(i)}, a(\cI)/W\}$. Recall that, by Proposition~\ref{prop:LB_OPT}, $M\leq OPT$.
Define $\pi = 2M$, and observe that
$$ 
(\pi-h_{\max}(\cI''))(W-w_{\max}(\cI''))+h_{\max}(\cI'')\cdot w_{\max}(\cI'')\ge M\cdot (W/2) + M\cdot (W/2) =M\cdot W \ge a(\cI).
$$ 
Thus we can apply Lemma~\ref{lem:packing-narrow-sorted} with parameter $\pi=2M$. This provides a schedule with peak at most $2M \le 2OPT$.
\end{proof}

In the following sections\kam{,} we will extend the approach in the above $2$-approximation as follows. We will first compute a feasible solution of some given peak that includes all the tasks having height larger than some threshold and width larger than some threshold. Then we will left-push this schedule to add some structure to the demand profile. Finally, we schedule the remaining tasks by means of a generalization of Lemma \ref{lem:packing-narrow-sorted} (Lemma~\ref{lem:packing-narrow-Qtsorted}) which considers $(Q,t^*)$-sorted schedules (rather than just sorted ones).




\section{A ($5/3+\eps$)-Approximation for \DSP}\label{sec:approx-alg}

In this section we will prove Theorem \ref{thr:53}.
In order to attain the claimed result, we will provide first some useful definitions and preprocessing lemmas.

Let us assume that the optimal value $OPT$ is known to the algorithm (this assumption can be dropped by approximately guessing this value, introducing an extra $(1+\varepsilon)$ factor in the approximation). We start by classifying the tasks in the instance according to their widths and heights (see Figure~\ref{fig:classif}). Let $\mu, \delta$, $\mu<\delta\leq \eps$, be two constant parameters to be fixed later. We say that a task $i$ is:
\begin{itemize}
	\item \emph{tall} if $h(i)> \frac{2}{3}OPT$,
	\item \emph{large} if $h(i) \in (\delta OPT, \frac{2}{3}OPT]$ and $w(i)>\varepsilon W$,
	\item \emph{horizontal} if $h(i) \le \mu OPT$ and $w(i) > \varepsilon W$,
	\item \emph{narrow} if $h(i) \le \frac{2}{3} OPT$ and $w(i) \le \varepsilon W$, or
	\item \emph{medium} if $h(i) \in (\mu OPT, \delta OPT]$ and $w(i)>\varepsilon W$.
\end{itemize}
	
\begin{figure}
    \centering
		\resizebox{!}{150pt}{\begin{tikzpicture}
			\fill[pattern color = lightgray, pattern = north east lines] (0,0) rectangle (2,5.25);
			\fill[pattern color = lightgray, pattern = north west lines] (2,2.5) rectangle (6,5.25);
			\fill[color = lightgray!55] (0,5.25) rectangle (6,7.25);
			\fill[color = lightgray!55] (2,1.25) rectangle (6,2.5);
			\fill[pattern color = lightgray, pattern = north west lines] (2,0) rectangle (6,1.25);
			
			
			\draw[dashed] (2,0) -- (2,7.75);
			\draw[dashed] (6,0) -- (6,7.75);
			
			\draw[dashed] (0,1.25) -- (6.5,1.25);
			\draw[dashed] (0,2.5) -- (6.5,2.5);
			\draw[dashed] (0,5.25) -- (6.5,5.25);
			\draw[dashed] (0,7.25) -- (6.5,7.25);
			
			
			\draw[->] (0,0) -- (6.5,0);
			\draw[->] (0,0) -- (0,7.75);
			
			
			\draw[ultra thick] (0,0) rectangle (2,5.25);
			\draw[ultra thick] (2,0) rectangle (6,1.25);
			\draw[ultra thick] (2,1.25) rectangle (6,2.5);
			\draw[ultra thick] (2,2.5) rectangle (6,5.25);
			\draw[ultra thick] (0,5.25) rectangle (6,7.25);
			
			
			\draw (0,0) -- (0,-0.25);
			\draw (0,-0.25) node[anchor=north] {$0$};
			\draw (2,0) -- (2,-0.25);
			\draw (2,-0.25) node[anchor=north] {$\varepsilon W$};
			\draw (6,0) -- (6,-0.25);
			\draw (6,-0.25) node[anchor=north] {$W$};
			
			\draw (0,0) -- (-0.25,0);
			\draw (-0.25,0) node[anchor=east] {$0$};
			\draw (0,1.25) -- (-0.25,1.25);
			\draw (-0.25,1.25) node[anchor=east] {$\mu OPT$};
			\draw (0,2.5) -- (-0.25,2.5);
			\draw (-0.25,2.5) node[anchor=east] {$\delta OPT$};
			\draw (0,5.25) -- (-0.25,5.25);
			\draw (-0.25,5.25) node[anchor=east] {$\frac{2}{3} OPT$};
			\draw (0,7.25) -- (-0.25,7.25);
			\draw (-0.25,7.25) node[anchor=east] {$OPT$};
			
			
			\draw (1,0.625) node {\textbf{\small Narrow}};
			\draw (1,1.875) node {\textbf{\small Narrow}};
			\draw (1,3.825) node {\textbf{\small Narrow}};
			\draw (1,6.25) node {\textbf{\small Tall}};
			
			
			\draw (4,0.625) node {\textbf{\small Horizontal}};
			\draw (4,1.875) node {\textbf{\small Medium}};
			\draw (4,3.825) node {\textbf{\small Large}};
			\draw (4,6.25) node {\textbf{\small Tall}};
			
			\draw (0,7.75) node[anchor=south] {$h(i)$};
			\draw (6.5,0) node[anchor=west] {$w(i)$};
			\end{tikzpicture}
		}
	\hspace{40pt}
		\resizebox{!}{150pt}{\begin{tikzpicture}
			\fill[pattern color = lightgray, pattern = north east lines] (0,0) rectangle (1,5);
            \fill[pattern color = lightgray, pattern = north east lines] (1,0) rectangle (1.5,6);
            \fill[pattern color = lightgray, pattern = north east lines] (1.5,0) rectangle (3,4.65);
            \fill[pattern color = lightgray, pattern = north east lines] (3,0) rectangle (4,5);
            \fill[pattern color = lightgray, pattern = north east lines] (4,0) rectangle (4.5,4);
            \fill[pattern color = lightgray, pattern = north east lines] (4.5,0) rectangle (5,3);
            \fill[pattern color = lightgray, pattern = north east lines] (5,0) rectangle (6,2);
            
	        \draw[ultra thick] (0,5) -- (1,5) -- (1,6) -- (1.5,6) -- (1.5,4.65) -- (3,4.65) -- (3,5) -- (4,5) -- (4,4) -- (4.5,4) -- (4.5,3) -- (5,3) -- (5,2) -- (6,2) -- (6,0);
            
			
			
			
			
			
			\draw[->] (0,0) -- (6.5,0);
			\draw[->] (0,0) -- (0,7.75);
			
			
			
			
			\draw (0,0) -- (0,-0.25);
			\draw (0,-0.25) node[anchor=north] {$0$};
			\draw (4,0) -- (4,-0.25);
			\draw (4,-0.25) node[anchor=north] {$t^*$};
			\draw (6,0) -- (6,-0.25);
			\draw (6,-0.25) node[anchor=north] {$W$};
			
			\draw (0,0) -- (-0.25,0);
			\draw (-0.25,0) node[anchor=east] {$0$};
			\draw (0,4.6) -- (-0.25,4.6);
			\draw (-0.25,4.6) node[anchor=east] {$Q$};
			
			\draw [dashed] (0,4.6) -- (6,4.6);
			\draw [dashed] (4,0) -- (4,5);
			
			
			

			\draw (0,7.75) node[anchor=south] {Demand};
			\draw (6.5,0) node[anchor=west] {Edges};
			\end{tikzpicture}
		}
	\caption{(Left) Classification of the tasks according to Section~\ref{sec:approx-alg}. (Right) Representation of a $(Q,t^*)$-sorted demand profile.}\label{fig:classif}
		\end{figure}
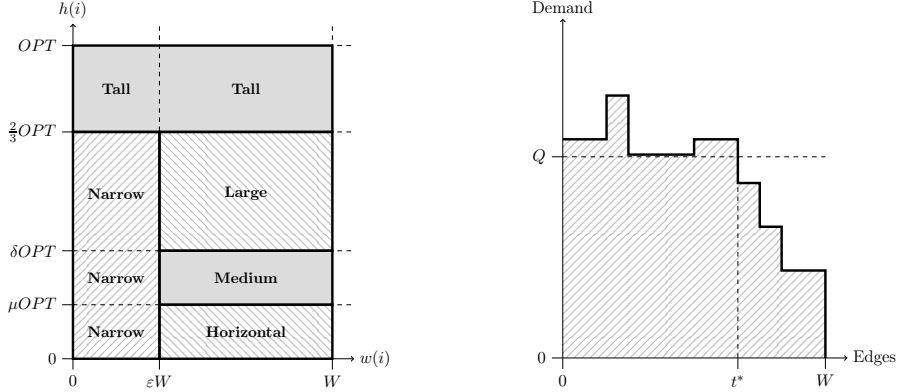

We denote by $\mathcal{T}, \mathcal{L}, \mathcal{H}, \mathcal{N}$ and $\mathcal{M}$ the sets of tall, large, horizontal, narrow and medium tasks respectively. As the following lemma states, it is possible to choose $\mu$ and $\delta$ in such a way that the two parameters differ by a large factor and the total height of medium tasks is small.
	
\begin{lemma}\label{lem:mediumrectanglesarea}
	Given a polynomial-time computable function $f : (0, 1) \rightarrow (0, 1)$, with $f(x) < x$, and any constant $\varepsilon\in (0,1)$, we can compute in polynomial time a set $\Delta$ of $\frac{2}{\varepsilon^2}$ many positive real numbers upper bounded by $\varepsilon$, such that there is at least one number $\delta \in \Delta$ so that, by choosing
	$\mu = f(\delta)$, one has 
	$a(\mathcal{M})\leq \varepsilon^2 \cdot OPT \cdot W$ (hence $h(\mathcal{M}) \le \varepsilon OPT$).		
\end{lemma}

\begin{proof} Let $y_1 = \varepsilon$ and, for each $j \in \{1,\dots,|\Delta|\}$, define $y_{j+1} = f(y_j)$. For each $j \leq |\Delta|$, let $I_j = \{i \in \mathcal{I} : h(i) \in [y_{j+1}, y_j)\}$. Note that $y_j$'s are decreasing since $f(x) < x$. Observe that $I_{j'}$ is disjoint from $I_{j''}$ for every $j' \neq j''$, and the total area of tasks in $\bigcup I_j$ is at most $W \cdot OPT$. Thus, there exists a value $\overline{j}$ such that the total area of the tasks in $I_{\overline{j}}$ is at most $\frac{2 OPT\cdot W}{|\Delta|} = \varepsilon^2 \cdot OPT\cdot W$. Choosing $\delta = y_{\overline{j}}$ and $\mu = y_{\overline{j}+1}$ verifies all the conditions of the lemma as in that case $\mathcal{M}\subseteq I_{\overline{j}}$. Notice that, since every task in $\mathcal{M}$ has width at least $\varepsilon W$, we have that $h(\mathcal{M}) \le \varepsilon OPT$.
\end{proof}

Function $f$ will be given later. From now on, we will assume that $\mu$ and $\delta$ are chosen according to Lemma~\ref{lem:mediumrectanglesarea}. \fab{Notice that this implies that $\mu,\delta=O_\eps(1)$}. The rest of this section is organized as follows. In Section \ref{sec:containerTL} we define a container-based scheduling of $\cT\cup \cL$. In Section \ref{sec:containerH} we extend this in order to include also $\cH$. In Section \ref{sec:schedulingALL} we schedule the remaining tasks and prove Theorem \ref{thr:53}.

\subsection{Containers for Tall and Large Tasks}
\label{sec:containerTL}

In this section\kam{,} we define a packing of tall and large tasks into a constant number of guessable containers. This packing can be computed exactly, i.e. with no leftovers (in particular, we will not use Lemma \ref{lem:containersPackPTAS} to compute such packing). To that aim\kam{,} we will exploit the following structural result.
\begin{lemma}\label{lem:structural_5_3} Let $P(\cdot)$ be an optimal schedule of $\cI$ (hence with peak $OPT$). There exists a packing $P'(\cdot)$ with peak at most $\frac{5}{3}OPT$ satisfying that all the tall tasks are scheduled one after the other starting on the leftmost edge in non-increasing order of height. \end{lemma}
\begin{proof} 
Let $\cT$ be the tall tasks (having height larger than $\frac{2}{3}OPT$). Notice that the paths of these tasks in $P(\cdot)$ need to be edge disjoint. Let us classify the edges into \emph{valley edges} if some task in $\mathcal{T}$ uses that edge in $P(\cdot)$ and \emph{mountain edges} otherwise (see also Figure~\ref{fig:structural_5_3A}). We let $\cI_{mnt}$ be the (\emph{mountain}) tasks whose path in $P(\cdot)$ consists solely of mountain edges, $\cI_{vll}$ be the (\emph{valley}) tasks whose path in $P(\cdot)$ consists solely of valley edges (notice that this set includes $\cT$), and $\cI_{crs}:=\cI\setminus (\cI_{vll}\cup \cI_{mnt})$ the remaining (\emph{crossing}) tasks.

We next define a modified partial schedule $P'(\cdot)$ of $\cI_{vll}\cup \cI_{mnt}$ as follows. Let us reorder the edges of the path (and the tasks accordingly) so that valley edges appear to the left and mountain edges appear to the right in the path (maintaining their relative order). Furthermore\kam{,} we rearrange the valley edges so that \fab{tasks in} $\mathcal{T}$ \fab{are scheduled from left to right in non-increasing order of height}.
Observe that by construction $\cI_{mnt}$ are scheduled on $W-w(\mathcal{T})$ edges (i.e. the total number of mountain edges). Since we temporarily removed crossing tasks, this induces a feasible packing of $\cI_{vll}\cup \cI_{mnt}$. 
The resulting packing $P'(\cdot)$ clearly has \kam{a} peak \kam{of} at most $OPT$.

Consider next the schedule $P(\cdot)$ \emph{restricted to} $\cI_{crs}$. We claim that this schedule has peak at most $\frac{2}{3}OPT$.  
%
%
%
%
Indeed, notice first that the demand of valley edges is at most $\frac{1}{3}OPT$. 
Consider next a mountain edge $e$. Let $e_{\ell}$ be the rightmost valley edge to the left of $e$ (if any), and define $e_r$ symmetrically to the right of $e$. Any task in $\cI_{crs}$ using $e$ must also use $e_{\ell}$ or $e_{r}$ (or both). Hence the total demand on $e$ is at most the total demand on $e_{\ell}$ plus the total demand on $e_{r}$, thus at most $\frac{2}{3}OPT$. The claim follows by combining the schedule of $\cI_{crs}$ (taken from $P(\cdot)$) with the above schedule $P'(\cdot)$ of $\cI_{vll}\cup \cI_{mnt}$ (see also Figure~\ref{fig:structural_5_3B}). 
\end{proof}

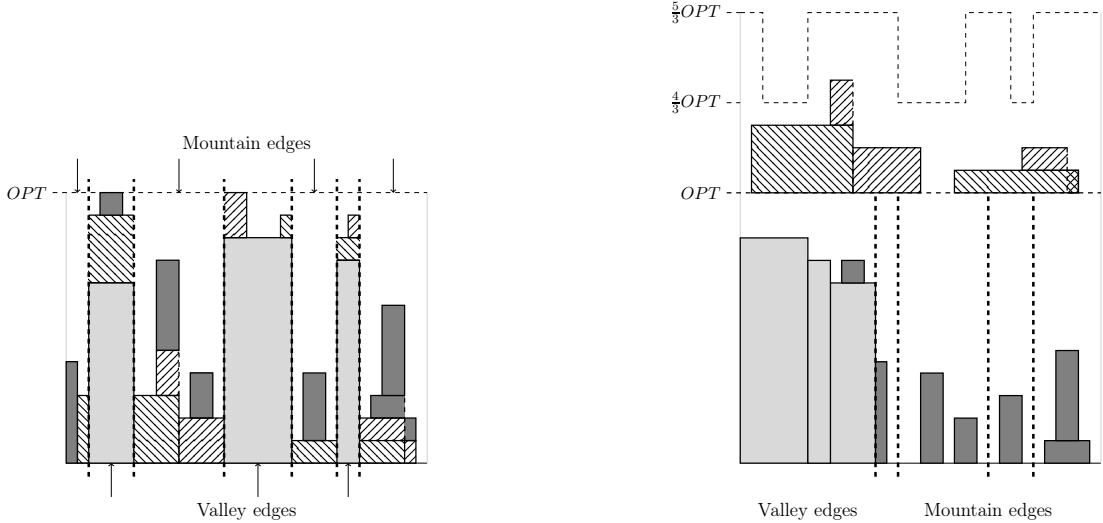
\begin{figure}
	\centering
	\begin{subfigure}{0.47\textwidth}
	\centering
	\resizebox{0.72\textwidth}{!}{\begin{tikzpicture}
	\draw (-0.3,10) node[anchor=east] {\color{white} $\frac{5}{3}OPT$};
	\draw (0,0) -- (8,0);
	\draw[lightgray!70] (0,0) -- (0,6);
	\draw[lightgray!70] (8,0) -- (8,6);
	
	
	\draw[thick, fill=lightgray!60] (0.5,0) rectangle (1.5,4);
	\draw[thick, fill=lightgray!60] (3.5,0) rectangle (5,5);
	\draw[thick, fill=lightgray!60] (6,0) rectangle (6.5,4.5);
	
	
	\draw[thick, fill=gray] (2,2.5) rectangle (2.5,4.5);
	\draw[thick, fill=gray] (2.75,1) rectangle (3.25,2);
	\draw[thick, fill=gray] (5.25,0.5) rectangle (5.75,2);
	\fill[gray] (6.75,1) rectangle (7.5,1.5);
	\fill[gray] (7.5,0.5) rectangle (7.75,1);
	\draw[thick, dashed] (7.5,0.5) -- (7.5,1.5);
	\draw[thick] (7.5,0.5) -- (7.75,0.5) -- (7.75,1) -- (7.5,1);
	\draw[thick] (7.5,1) -- (6.75,1) -- (6.75,1.5) -- (7.5,1.5);
	\draw[thick, fill=gray] (7,1.5) rectangle (7.5,3.5);
	\draw[thick, fill=gray] (0,0) rectangle (0.25,2.25);
	
	
	\fill[pattern=north west lines] (0.25,0) rectangle (0.5,1.5);
	\fill[pattern=north west lines] (0.5,4) rectangle (1.5,5.5);
	\fill[pattern=north west lines] (1.5,0) rectangle (2.5,1.5);
	\draw[thick] (0.5,0) -- (0.25,0) -- (0.25,1.5) -- (0.5,1.5);
	\draw[thick] (0.5,5.5) -- (1.5,5.5);
	\draw[thick] (0.5,4) -- (1.5,4);
	\draw[thick] (1.5,0) -- (2.5,0) -- (2.5,1.5) -- (1.5,1.5);
	\fill[pattern=north east lines] (2.5,0) rectangle (3.5,1);
	\fill[pattern=north east lines] (2,1.5) rectangle (2.5,2.5);
	\fill[pattern=north east lines] (3.5,5) rectangle (4,6);
	\draw[thick] (2.5,0) -- (3.5,0);
	\draw[thick] (3.5,1) -- (2.5,1);
	\draw[thick] (3.5,5) -- (4,5) -- (4,6) -- (3.5,6);
	\draw[thick] (2.5,1.5) -- (2,1.5) -- (2,2.5) -- (2.5,2.5);
	\draw[thick, dashed] (2.5,0) -- (2.5,1);
	\draw[thick, dashed] (2.5,1.5) -- (2.5,2.5);
	\fill[pattern=north west lines] (4.75,5) rectangle (5,5.5);
	\fill[pattern=north west lines] (5,0) rectangle (6,0.5);
	\fill[pattern=north west lines] (6,4.5) rectangle (6.5,5);
	\fill[pattern=north west lines] (6.5,0) rectangle (7.5,0.5);
	\draw[thick] (5,5) -- (4.75,5) -- (4.75,5.5) -- (5,5.5);
	\draw[thick] (5,0) -- (6,0);
	\draw[thick] (5,0.5) -- (6,0.5); 
	\draw[thick] (6,4.5) -- (6.5,4.5);
	\draw[thick] (6,5) -- (6.5,5);
	\draw[thick] (6.5,0) -- (7.5,0) -- (7.5,0.5) -- (6.5,0.5);
	\draw[thick, fill=gray] (0.75,5.5) rectangle (1.25,6);
	\fill[pattern=north east lines] (6.25,5) rectangle (6.5,5.5);
	\fill[pattern=north east lines] (6.5,0.5) rectangle (7.5,1);
	\fill[pattern=north east lines] (7.5,0) rectangle (7.75,0.5);
	\draw[thick] (6.5,5) -- (6.25,5) -- (6.25,5.5) -- (6.5,5.5);
	\draw[thick] (6.5,1) -- (7.5,1);
	\draw[thick,dashed] (7.5,0.5) -- (7.5,1);
	\draw[thick] (7.5,0) -- (7.75,0) -- (7.75,0.5) -- (7.5,0.5);
	
	\draw[ultra thick, dashed] (0.5,-0.3) -- (0.5,6.3);
	\draw[ultra thick, dashed] (1.5,-0.3) -- (1.5,6.3);
	\draw[ultra thick, dashed] (3.5,-0.3) -- (3.5,6.3);
	\draw[ultra thick, dashed] (5,-0.3) -- (5,6.3);
	\draw[ultra thick, dashed] (6,-0.3) -- (6,6.3);
	\draw[ultra thick, dashed] (6.5,-0.3) -- (6.5,6.3);
	
	\draw[->] (0.25,6.75) -- (0.25,6);
	\draw[->] (2.5,6.75) -- (2.5,6);
	\draw[->] (5.5,6.75) -- (5.5,6);
	\draw[->] (7.25,6.75) -- (7.25,6);
	\draw (4,6.75) node[anchor=south] {\large Mountain edges};
	\draw[->] (1,-0.75) -- (1,0);
	\draw[->] (4.25,-0.75) -- (4.25,0);
	\draw[->] (6.25,-0.75) -- (6.25,0);
	\draw (4,-0.75) node[anchor=north] {\large Valley edges};
	\draw[dashed] (-0.3,6) -- (8,6);
	\draw (-0.3,6) node[anchor=east] {$OPT$};
	\end{tikzpicture}}
\caption{A scheduling of peak $OPT$. Light gray rectangles correspond to tasks in $\mathcal{T}$, dark gray ones represent mountain and valley tasks, and dashed ones the crossing tasks.}\label{fig:structural_5_3A}
	\end{subfigure}
	\qquad
	\begin{subfigure}{0.47\textwidth}
	\centering
	\resizebox{0.72\textwidth}{!}{\begin{tikzpicture}
	\draw (0,0) -- (8,0);
	\draw[lightgray!70] (0,0) -- (0,10);
	\draw[lightgray!70] (8,0) -- (8,10);
	
	
	\draw[thick, fill=lightgray!60] (2,0) rectangle (3,4);
	\draw[thick, fill=lightgray!60] (0,0) rectangle (1.5,5);
	\draw[thick, fill=lightgray!60] (1.5,0) rectangle (2,4.5);
	
	
	\draw[thick, fill=gray] (4,0) rectangle (4.5,2);
	\draw[thick, fill=gray] (4.75,0) rectangle (5.25,1);
	\draw[thick, fill=gray] (5.75,0) rectangle (6.25,1.5);
	\draw[thick, fill=gray] (6.75,0) rectangle (7.75,0.5);
	\draw[thick, fill=gray] (7,0.5) rectangle (7.5,2.5);
	\draw[thick, fill=gray] (3,0) rectangle (3.25,2.25);
	
	
	\draw[thick, pattern=north west lines] (0.25,6) rectangle (2.5,7.5);
	\fill[pattern=north east lines] (2.5,6) rectangle (4,7);
	\fill[pattern=north east lines] (2,7.5) rectangle (2.5,8.5);
	\draw[thick] (2.5,6) -- (4,6) -- (4,7) -- (2.5,7);
	\draw[thick] (2.5,7.5) -- (2,7.5) -- (2,8.5) -- (2.5,8.5);
	\draw[thick, dashed] (2.5,6) -- (2.5,7);
	\draw[thick, dashed] (2.5,7.5) -- (2.5,8.5);
	\draw[thick, pattern=north west lines] (4.75,6) rectangle (7.5,6.5);
	\draw[thick, fill=gray] (2.25,4) rectangle (2.75,4.5);
	\fill[pattern=north east lines] (6.25,6.5) rectangle (7.25,7);
	\fill[pattern=north east lines] (7.5,6) rectangle (7.25,6.5);
	\draw[thick, dashed] (7.25,7) -- (7.25,6);
	\draw[thick] (7.25,6.5) -- (6.25,6.5) -- (6.25,7) -- (7.25,7);
	\draw[thick] (7.25,6) -- (7.5,6) -- (7.5,6.5) -- (7.25,6.5);
	
	\draw[ultra thick, dashed] (3,-0.3) -- (3,6);
	\draw[ultra thick, dashed] (3.5,-0.3) -- (3.5,6);
	\draw[ultra thick, dashed] (5.5,-0.3) -- (5.5,6);
	\draw[ultra thick, dashed] (6.5,-0.3) -- (6.5,6);
	
	\draw (5.5,-0.75) node[anchor=north] {\large Mountain edges};
	\draw (1.5,-0.75) node[anchor=north] {\large Valley edges};
	\draw[dashed, thick] (-0.3,6) -- (8,6);
	\draw (-0.3,6) node[anchor=east] {$OPT$};
	\draw[dashed] (0,10) -- (0.5,10) -- (0.5,8) -- (1.5,8) -- (1.5,10) -- (3.5,10) -- (3.5,8) -- (5,8) -- (5,10) -- (6,10) -- (6,8) -- (6.5,8) -- (6.5,10) -- (8,10);
	\draw[dashed] (-0.3,8) -- (0,8);
	\draw (-0.3,8) node[anchor=east] {$\frac{4}{3}OPT$};
	\draw[dashed] (-0.3,10) -- (0,10);
	\draw (-0.3,10) node[anchor=east] {$\frac{5}{3}OPT$};
	\end{tikzpicture}}
\caption{Structured solution having peak at most $\frac{5}{3}OPT$, where tasks in $\mathcal{T}$ are placed one next to the other, starting at the leftmost edge and sorted non-increasingly by height.}\label{fig:structural_5_3B}
\end{subfigure}
	\caption{Depiction of the proof of Lemma~\ref{lem:structural_5_3}.}\label{fig:structural_5_3}
\end{figure}

%

We will next assume that tall tasks are scheduled as in the above lemma. By increasing the peak by $\eps OPT$ (up to $(\frac{5}{3}+\eps)OPT$), one can define a set of $O_{\eps}(1)$ (tall) containers where such tasks can be packed (respecting the mentioned order and with no leftovers). Consider the demand profile of tall tasks in the considered schedule, and round it to the next multiple of $\eps OPT$. Consider the tasks $\cT_k$ corresponding to the value $k\cdot \eps OPT$ in the rounded profile. Notice that these tasks are scheduled consecutively along some path $P_k$. We create a vertical container $C_k$ of height $k\cdot \eps OPT$ and width $|E(P_k)|$, pack $\cT_k$ into $C_k$, and schedule $C_k$ on $P_k$. Clearly\kam{,} we need to create at most $1/\eps$ containers. Notice also that the dimensions of these containers can be guessed in polynomial time since there is a constant number of options for the height, and the widths correspond to the total width of a subsequence of tall tasks in the considered ordering by non-increasing height (breaking ties arbitrarily).

It remains to consider large tasks. Since they are at most $\frac{1}{\eps \delta}=O_{\eps}(1)$ many, it is sufficient to define a distinct (large) container for each one of them and pack the large tasks accordingly. We schedule the large containers exactly as in the solution guaranteed by Lemma \ref{lem:structural_5_3}. Clearly tall and large containers can be scheduled together with \kam{a} peak \kam{of} at most $(\frac{5}{3}+\eps)OPT$ by the above construction. 


\subsection{Containers for Horizontal Tasks}
\label{sec:containerH}

In this section\kam{,} we define a packing of horizontal tasks into a constant number of guessable containers. These containers can be scheduled together with the tall and large containers with small enough peak\kam{s}. This will induce a convenient schedule of non-narrow tasks. 

Let us focus on the schedule of tall and large containers with \kam{a} peak \kam{of} at most $(\frac{5}{3}+\eps)OPT$ from the previous section. Consider now the demand profile of such container schedule. Since the demand profile of tall containers has at most $1/\varepsilon=O_{\eps}(1)$ jumps, and the demand profile of large containers has at most $2/(\varepsilon \delta)=O_{\eps}(1)$ jumps,\fabr{I added a remark before that $\mu,\delta$ are $O_{\eps}(1)$, check!} then the overall demand profile has $O_{\eps}(1)$ jumps. 

Assume next that horizontal tasks are scheduled as in Lemma \ref{lem:structural_5_3}: notice that such tasks can be scheduled with the tall and large containers without increasing the peak. This implies that the demand profile of horizontal tasks is upper bounded (on each coordinate) by the difference between $(5/3+\varepsilon)OPT$ and the demand profile of tall and large containers (see Figure~\ref{fig:lemma11}). Under these conditions, it is possible to build containers for horizontal tasks, using the standard linear grouping technique, as the following lemma shows. 

\begin{figure}
\centering
	\resizebox{0.47\textwidth}{!}{\begin{tikzpicture} 
	\draw (6.3,10) node[anchor=west] {\textcolor{white}{\Large $\left(\frac{5}{3}+\varepsilon\right)OPT$}};
	\draw (0,0) -- (6,0);
	\draw[lightgray!60] (0,0) -- (0,10);
	\draw[lightgray!60] (6,0) -- (6,10);
	
	\draw[fill=gray!50] (0,0) rectangle (1,5.75);
	\draw[fill=gray!40] (1,0) rectangle (1.5,5.25);
	\draw[fill=gray!30] (1.5,0) rectangle (2.5,4.75);
	\draw[fill=gray!20] (2.5,0) rectangle (3,4.25);
	
	\fill[pattern=north east lines, pattern color=gray!90] (0,5.75) rectangle (0.5,7.5);
	\fill[pattern=north east lines, pattern color=gray!90] (0.5,5.75) rectangle (1,7);
	\fill[pattern=north east lines, pattern color=gray!90] (1,5.25) rectangle (1.5,7);
	\fill[pattern=north east lines, pattern color=gray!90] (1.5,4.75) rectangle (2,8);
	\fill[pattern=north east lines, pattern color=gray!90] (2,4.75) rectangle (2.5,7.5);
	\fill[pattern=north east lines, pattern color=gray!90] (2.5,4.25) rectangle (3,6);
	\fill[pattern=north east lines, pattern color=gray!90] (3,0) rectangle (3.5,2.5);
	\fill[pattern=north east lines, pattern color=gray!90] (3.5,0) rectangle (4.5,1.5);
	\fill[pattern=north east lines, pattern color=gray!90] (4.5,0) rectangle (5.5,1);
	\fill[pattern=north east lines, pattern color=gray!90] (5.5,0) rectangle (6,3);
	
	\fill[pattern=horizontal lines, pattern color=gray!30] (0,8.25) rectangle (0.5,10);
	\fill[pattern=horizontal lines, pattern color=gray!30] (0.5,7.75) rectangle (0.75,10);
	\fill[pattern=horizontal lines, pattern color=gray!30] (0.75,7.5) rectangle (1.25,10);
	\fill[pattern=horizontal lines, pattern color=gray!30] (1.25,9.25) rectangle (2,10);
	\fill[pattern=horizontal lines, pattern color=gray!30] (2,8.25) rectangle (2.25,10);
	\fill[pattern=horizontal lines, pattern color=gray!30] (2.25,9) rectangle (2.75,10);
	\fill[pattern=horizontal lines, pattern color=gray!30] (2.75,7) rectangle (3.5,10);
	\fill[pattern=horizontal lines, pattern color=gray!30] (3.5,5) rectangle (4,10);
	\fill[pattern=horizontal lines, pattern color=gray!30] (4,2.75) rectangle (4.75,10);
	\fill[pattern=horizontal lines, pattern color=gray!30] (4.75,4) rectangle (5.25,10);
	\fill[pattern=horizontal lines, pattern color=gray!30] (5.25,6.5) rectangle (6,10);
	
	\draw[ultra thick] (0,5.75) -- (1,5.75) -- (1,5.25) -- (1.5,5.25) -- (1.5,4.75) -- (2.5,4.75) -- (2.5,4.25) -- (3,4.25) -- (3,0) -- (0,0) -- (0,5.75);
	
	\draw[ultra thick] (0,5.75) -- (0,7.5) -- (0.5,7.5) -- (0.5,7) -- (1.5,7) -- (1.5,8) -- (2,8) -- (2,7.5) -- (2.5,7.5) -- (2.5,6) -- (3,6) -- (3,0);
	\draw[ultra thick] (3,2.5) -- (3.5,2.5) -- (3.5,1.5) -- (4.5,1.5) -- (4.5,1) -- (5.5,1) -- (5.5,3) -- (6,3) -- (6,0) -- (3,0) -- (3,2.5);
	
	\draw[ultra thick] (6,6.5) -- (6,10) -- (0,10) -- (0,8.25) -- (0.5,8.25) -- (0.5,7.75) -- (0.75,7.75) -- (0.75,7.5) -- (1.25,7.5) -- (1.25,9.25) -- (2,9.25) -- (2,8.25) -- (2.25,8.25) -- (2.25,9) -- (2.75,9) -- (2.75,7) -- (3.5,7) -- (3.5,5) -- (4,5) -- (4,2.75) -- (4.75,2.75) -- (4.75,4) -- (5.25,4) -- (5.25,6.5) -- (6,6.5);
	
	\draw[dashed] (-0.3,10) -- (0,10);
	\draw (-0.3,10) node[anchor=east] {\Large $\left(\frac{5}{3}+\varepsilon\right)OPT$};
	
	\draw (3,9.625) node {\large \textbf{Horizontal}};
	\draw (1.41,3) node {\large \textbf{Tall}};
	\draw (1.375,6.375) node {\large \textbf{Large}};
	\draw (4.5,0.625) node {\large \textbf{Large}};
	\end{tikzpicture}}
\caption{The demand profile of tall and large tasks in the schedule obtained from Lemma~\ref{lem:structural_5_3} has $O_{\varepsilon}(1)$ jumps, bounding the profile of (sliced) horizontal tasks (on top).
}\label{fig:lemma11}
\end{figure}
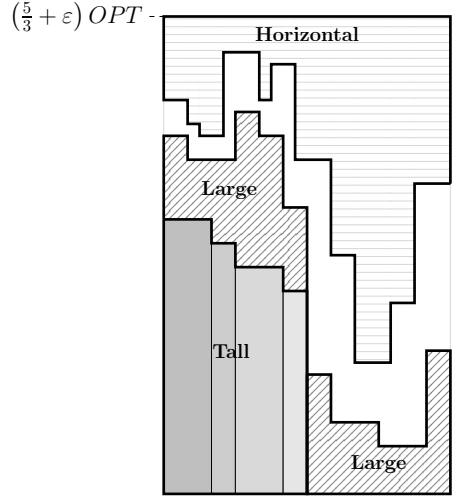

\begin{lemma}\label{lem:container-horizontal} Suppose there exists a schedule of $\cH$ 
such that its demand profile is upper bounded (vectorially) by a demand profile $D$ with $O_{\varepsilon}(1)$ jumps. Then there exists a container packing for $\mathcal{H}$ into $O_{\eps}(1)$ horizontal guessable containers with demand profile upper bounded by $D$ plus $4\varepsilon OPT$ on each coordinate.
\end{lemma}

\begin{proof}

Let us assume by now that horizontal tasks are \emph{horizontally sliced}, meaning that each task $i$, having height $h(i)$ and width $w(i)$, is replaced by $h(i)$ sibling slices, which are tasks of width $w(i)$ and height $1$. The schedule of the slices is the same as for the corresponding task. In order to reduce the possible number of distinct slice widths to a constant\kam{,} we will use the technique of linear grouping while increasing the final peak by at most $\wal{2}\varepsilon OPT$. We start by considering all the slices in a pile, 
one on top of the other and sorted non-increasingly by width from bottom to top (and putting sibling slices consecutively). Since these slices have \kam{a} width at least $\varepsilon W$ and total area at most $OPT\cdot W$, the pile has total height at most $\frac{1}{\varepsilon} OPT$. Starting from the bottom, we partition the pile into groups $G_1,\ldots,G_q$ of height exactly $\varepsilon OPT$ (except possibly for $G_q$ which may have \kam{a} smaller height). We remove from the solution the slices in $G_1$ and any slice in $G_2$ which used to have a sibling slice in $G_1$, and temporarily remove the corresponding tasks. Observe that we are removing the tasks whose slices are fully contained in $G_1$ plus at most one extra task. In particular, the total height of the removed tasks is at most $(\varepsilon + \mu)OPT \le 2\varepsilon OPT$ (here we use $\mu\leq \eps$).

Next\kam{,} we round up the widths of the remaining slices as follows: for $i=2,\ldots,q$, the width of slices (still) in $G_i$ are rounded to the smallest width of any slice originally in $G_{i-1}$. We call this set of slices the rounded slices, and next focus on packing them. Notice that rounded slices have at most $1/\varepsilon^2$ distinct widths. This also induces a matching between each rounded slice $a$ and a distinct original slice $b$, so that $w(b)\geq w(a)$. In particular, we can schedule each such $a$ starting on the first edge of $P(b)$ without increasing the overall peak.


Now we left-shift the horizontal slices in the solution as much as possible while still \fab{obtaining a schedule whose demand profile is} upper bounded by $D$. \fab{Let $e$ be the starting edge of some slice $S$ at the end of the process. Notice that one of the following cases holds: (1) $D$ increases on edge $e$ (including as a special case when $e$ is the leftmost edge of $G$) or (2) the edge $f$ to the left of $e$ is the ending edge of some other slice $S'$. Indeed otherwise it would be possible to left-shift $S$ while respecting all the constraints. This implies that the possible positions for the starting edge of any slice can be obtained by considering the $O_\eps(1)$ edges where $D$ increases and then adding the total width of a few slices. Notice that there are at most $1/\eps^2$ such widths, and we can sum up at most $1/\eps$ of them (since horizontal slices have width\kam{s} at least $\eps W$).} 
%
\fab{Altogether}, the number of possible starting edges for the slices is $O_{\eps}(1)$.

Consider the leftmost possible such edge $e$ and all the rounded slices $S_e$ starting on $e$ in this left-shifted schedule. We partition $S_e$ by width $w$, and for each such width $w$ and corresponding set of slices $S_{e,w}$, we construct a horizontal container of width $w$ and height $h(S_{e,w})$ where we pack $S_{e,w}$. We repeat this procedure for each possible starting edge $e$, obtaining in the end $K\le O_{\eps}(1)$ containers in total where we packed all the rounded slices. Notice that the width of each container is the width of some rounded slice, which in turn is the width of some task. Hence the widths of the containers are guessable in the usual sense.

We next turn the above packing of rounded slices into a feasible packing of tasks (into the same containers). First of all, we repack the rounded slices as follows. We consider all the slices $S_w$ of a given width $w$ in any order where sibling slices appear consecutively, and all the containers $C_w$ of that width in any order.  We pack each slice $s\in S_w$ in the first container $C\in C_w$ where $s$ still fits. Notice that $h(S_w)=h(C_w)$ by construction, hence we repack all rounded slices this way.  
Next we consider the tasks $i$ whose slices are all contained in the \emph{same} container $C$, and pack $i$ into $C$. By construction this packing is feasible. We add the tasks which are not packed this way to the set of removed tasks defined earlier. 
We round up the heights of the containers to the next multiple of $\frac{\varepsilon}{K}OPT$, hence making such heights guessable. This way the peak increases at most by $\varepsilon OPT$.

Consider the set of removed tasks. Recall that the tasks removed in the initial rounding phase have total height at most $2\eps OPT$. In the following packing phase\kam{,} we remove at most one task per container, hence these removed tasks have height at most $K\cdot \mu\cdot OPT \le \varepsilon OPT$. \fab{Here we assume that $\mu\leq \eps/K$: this can be achieved by choosing $f(x)=x/K$ in Lemma ~\ref{lem:mediumrectanglesarea}.}\fabr{Check my f}
Hence the removed tasks altogether have total height at most $3\varepsilon OPT$: we pack these tasks in one extra (guessable) horizontal container of width $W$ and height $3\varepsilon OPT$. 
\end{proof}

From the above construction it is possible to derive a schedule of non-narrow tasks of small enough peak.
\begin{lemma}\label{lem:packing_notnarrow} 
It is possible to compute in polynomial-time a feasible schedule of $\mathcal{I}\setminus\mathcal{N}$ with peak at most $\left(\frac{5}{3}+7\varepsilon\right)OPT$. 
\end{lemma}
\begin{proof}
Consider the guessable containers for tall, large\kam{,} and horizontal tasks and the corresponding schedule as described before. This schedule has \kam{a} peak \kam{of} at most $(\frac{5}{3}+\eps)OPT+4\eps OPT$. By Lemma \ref{lem:mediumrectanglesarea} the medium tasks fit into \kam{a}\kamr{changed ``an'' to ``a''} horizontal container of width $W$ and height $\eps OPT$. Altogether this leads to a packing of $\cL\cup \cT\cup \cH \cup \cM=\cI\setminus \cN$ into $O_\eps(1)$ guessable containers that can be scheduled with peak at most $(\frac{5}{3}+6\eps)OPT$. 

It is easy to pack $\cL\cup \cT\cup \cM$ into the corresponding containers. For $\cH$ we apply Lemma~\ref{lem:containersPackPTAS} with $\varepsilon' = \varepsilon^2$ to assign the horizontal tasks to them, obtaining a set of horizontal unplaced tasks of \kam{an} area at most $\varepsilon^2 \cdot W \cdot OPT$ (hence of total height at most $\eps OPT$). The latter tasks can be placed into an extra horizontal container of height $\varepsilon OPT$ and width $W$. The resulting set of containers can be scheduled with \kam{a} peak at most $(\frac{5}{3}+7\eps)OPT$, and such a schedule can be efficiently computed as already discussed.
\end{proof}

\subsection{Scheduling Narrow Tasks}
\label{sec:schedulingALL}

At this point it just remains to schedule the narrow tasks. For this goal we need the following generalization of Lemma \ref{lem:packing-narrow-sorted} that considers $(Q,t^*)$-sorted partial schedules. Recall that for a node $t^*$ of the path and a value $Q\geq 0$, a schedule is $(Q,t^*)$-sorted if the demand to the left of $t^*$ is at least $Q$, and to the right of $t^*$ the demand profile is non-increasing. 
\begin{lemma}\label{lem:packing-narrow-Qtsorted} Let $\cI$ be an instance of \DSP and $\alpha>0$. Suppose we are given a $((1+\alpha)OPT,t^*)$-sorted schedule of $\cI' \subseteq \cI$. \fab{Let} $\cI'':=\cI\setminus \cI'$ \fab{and} assume that:
\begin{itemize}
    \item \fab{The peak of the schedule} is at most $\pi$, with $\pi\ge (1+\alpha)OPT + h_{\max}(\cI'')$, and
    \item $w_{\max}(\cI'') \le \frac{\alpha}{2(\alpha+1)}W$.
\end{itemize}
Then it is possible to compute in polynomial time a schedule of $\cI$ with peak at most $\pi$. 
\end{lemma} 
\begin{proof}
By overloading notation, let $t^*$ also denote the number of edges to the left of $t^*$.
Notice first that $W-t^*\ge 2w_{\max}(\cI'')$. Indeed otherwise, since the input schedule is $((1+\alpha)OPT,t^*)$-sorted, the total area of the tasks in $\cI'$ would be at least 
$$
t^* \cdot (1+\alpha)OPT > (W-2w_{\max}(\fab{\cI''}))(1+\alpha)OPT \ge W\cdot OPT
$$ 
which is not possible. \kam{Roughly speaking, to} prove the desired claim\kam{,} we will apply Lemma~\ref{lem:packing-narrow-sorted} to the demand profile induced by the edges to the right of $t^*$. In more detail, we consider a new instance defined by a path with $\tilde{W}=W-t^*$ edges and a set of tasks $\tilde{\cI}$ consisting of $\tilde{\cI}'':=\cI''$ plus a set $\tilde{\cI}'$ of $\tilde{W}$ tasks having width $1$ and, for each edge $e$ to the right of $t^*$, height equal to the total demand on edge $e$ in the original schedule for $\cI'$. To see that the required hypotheses are satisfied, notice that by scheduling the tasks in $\tilde{\cI}'$ one next to the other sorted non-increasingly by height we obtain a sorted partial schedule of \kam{a} peak at most $\pi$, where 
$$
\pi\ge h_{\max}(\cI'') + (1+\alpha)OPT\ge h_{\max}(\tilde{\cI}'') + \max\{a(\tilde{\cI})/\tilde{W}, h_{\max}(\tilde{\cI}'')\}.
$$ 
The last inequality above holds since $a(\tilde{\cI}) \le OPT\cdot W - a(\cI') \le OPT(W-(1+\alpha)t^*)$. Finally, we notice that $w_{\max}(\tilde{\cI}'') \le \tilde{W}/2$, and
\begin{eqnarray*} 
& & (\tilde{W} - w_{\max}(\tilde{\cI}''))(\pi-h_{\max}(\tilde{\cI}'')) + w_{\max}(\tilde{\cI}'')\cdot h_{\max}(\tilde{\cI}'') \\ & \ge & (\tilde{W}-w_{\max}(\tilde{\cI}''))(1+\alpha)OPT \\ & = & OPT(W - t^*(1+\alpha)) + \alpha W\cdot OPT - w_{\max}(\cI'')(1+\alpha)OPT \\ & \ge & a(\tilde{\cI}) + \alpha W\cdot OPT - \frac{\alpha}{2}W\cdot OPT \ge a(\tilde{\cI}).
\end{eqnarray*}
Hence all the conditions of Lemma~\ref{lem:packing-narrow-sorted} apply. Given that the demand profile of the partial schedule for $\tilde{\cI}'$ is the same as the demand profile induced by the edges to the right of $t^*$ in the original schedule for $\cI'$, we can schedule $\cI''$ on top of the input schedule without exceeding the peak $\pi$.
\end{proof}

We now have all the ingredients to prove Theorem \ref{thr:53}. 
\begin{proof}[Proof of Theorem \ref{thr:53}]
Consider the schedule of $\cI\setminus \cN$ with peak at most $\pi:= \left(\frac{5}{3}+7\varepsilon\right)OPT$ provided by Lemma \ref{lem:packing_notnarrow}. We perform a \wal{$\pi$}-left-pushing of this schedule, however without left-shifting any tall task. Let us prove that this partial schedule of $\cI':=\cI\setminus \cN$ satisfies all the required properties of Lemma~\ref{lem:packing-narrow-Qtsorted} with parameter $\pi$. First of all, there exists a node $t^*$ for which (1) every edge to the left of $t^*$ (if any) has demand larger than $(1+7\varepsilon)OPT$, and (2) the demand profile to the right of $t^*$ is non-increasing. Indeed, if (1) does not hold, then there exists an edge having demand less than $(1+7\varepsilon)OPT$ and the following edge has demand larger than $(1+7\varepsilon)OPT$. But this means that some task which is not tall can be left-shifted (as there can be only one tall task per edge); similarly, if (2) does not hold, there is a pair of contiguous edges to the right of $t^*$ where the demand strictly increases from left to right. But since the tall tasks are sorted non-increasingly by height, this implies that there exists a task \kam{that} \kamr{changed ``which'' to ``that''} is not tall that can be left-shifted. In conclusion, the solution is $((1+7\varepsilon)OPT,t^*)$-sorted and also $w_{\max}(\mathcal{N}) \le \varepsilon W \le \frac{7\varepsilon}{2(1+7\varepsilon)}W$ for $\varepsilon$ small enough. Since $\pi\geq (1+7\varepsilon)OPT + h_{\max}(\mathcal{N})$, by Lemma~\ref{lem:packing-narrow-Qtsorted} we obtain a feasible schedule of peak at most $\pi$. The claim follows by scaling $\eps$ appropriately.
\end{proof}

\bibliographystyle{abbrv}
\bibliography{references}


\newpage
\appendix

\section{Hardness of Approximation for \textsc{Square-DSP}}\label{sec:hardness}

For several rectangle packing problems it is usually the case that they are NP-hard (or APX-hard) even when restricted to instances consisting solely of squares~\cite{LTWYC90}. This holds also for DSP, as the following theorem shows.

\begin{theorem}
\label{thm:hardness_pecs}
For any $\varepsilon > 0$, there exists no polynomial-time $(3/2 - \varepsilon)$-approximation algorithm for \DSPS unless NP $=$ P.
\end{theorem}

In order to prove this result, we show a gap-producing reduction from the NP-complete \BP\ problem, formally defined as follows. 

\begin{definition}[\BP]
In an instance of the \BP\ problem, we are given a set of $2n$ positive integers $A = \{a_1, a_2, \ldots, a_{2n}\}$. The goal is to decide whether there exists a partitioning of $A$ into $A_1$ and $A_2$ such that $\abs{A_1} = \abs{A_2} = n$, and the sum of the numbers in each of the two sets is equal to a target value $B = \left(\sum_{j = 1}^{2n} a_j \right) / 2$.
\end{definition}

We start first by proving that the \BP\ problem is NP-complete. This result is folklore by now but, for the sake of completeness, we bring a complete proof.

\begin{theorem}
\BP\ is NP-complete.
\end{theorem}

\begin{proof}
We will reduce the \Part\ problem to the balanced variant in polynomial time. Given an instance $\cI$ of the \Part\ problem with $n$ numbers $a_1, a_2, \ldots, a_n$, we construct an instance $\cI'_k$ of the \BP\ problem for each $k \in \{1, 2, \ldots, \floor{n / 2} \}$. Let $C$ be $(\sum_{j = 1}^n a_j) + 1$. For each $k$, the instance $\cI'_k$ is defined as follows. Let $\cI'_k$ have all the initial numbers, $a_1, a_2, \ldots, a_n$. Add the set \textsf{dummy} of $n - 2k + 2$ extra numbers where $\textsf{dummy} = \{ \alpha, \beta_1, \beta_2, \ldots, \beta_{n - 2k + 1} \}$ in which $\alpha = (n - 2k + 1)C$ and $\beta_i = C$ for each $i \in [n - k + 1]$. We claim that $\cI$ is a \emph{Yes} instance of the \Part\ problem if and only if at least one $\cI'_k$ is a \emph{Yes} instance of the \BP\ problem.\\

\noindent\textbf{Completeness:} Assume that $\cI$ is a \emph{Yes} instance of the \Part\ problem. Let $S_1$ and $S_2$ be the two sets of equal sum, say $B$. These sets may not necessarily have the same cardinality. With sum renumbering, assume that $S_1 = \{a_1, a_2, \ldots, a_k \}$ and $S_2 = \{a_{k + 1}, a_{k + 2}, \ldots, a_n \}$ for some $k \in [\floor{n / 2}]$. It is easy to see that the instance $\cI'_k$ of the \BP\ problem is a \emph{Yes} instance, since we can make the sets $S'_1 = \{a_1, a_2, \ldots, a_k, \beta_1, \beta_2, \ldots, \beta_{n - 2k + 1}\}$ and $S'_2 = \{a_{k + 1}, a_{k + 2}, \ldots, a_n, \alpha \}$ both with the sum $B + (n - 2k + 1)C$ and cardinality $n - k + 1$.\\

\noindent\textbf{Soundness:} Now assume that $\cI$ is a \emph{No} instance of the \Part\ problem. We claim that no $\cI'_k$ can be a \emph{Yes} instance of the \BP\ problem either. For the sake of contradiction, assume $\cI'_k$ is a \emph{Yes} instance with two partitions $S'_1$ and $S'_2$ of the same sum and cardinality. Note that if no $\beta_i$ is placed in the same set as $\alpha$, we reach a contradiction since we then can find two sets $S_1 = \{a_1, a_2, \ldots, a_k \}$ and $S_2 = \{a_{k + 1}, a_{k + 2}, \ldots, a_n \}$ of the same sum of the original \Part\ instance $\cI$. So with some renumbering, we can assume we have $S'_1 = \{a_1, a_2, \ldots, a_{k'}, \alpha, \beta_1, \beta_2, \ldots, \beta_\ell \}$ and $S'_2 = \{ a_{k' + 1} + a_{k' + 2}, \ldots, a_n, \beta_{\ell + 1}, \beta_{\ell + 2}, \ldots, \beta_{n - 2k + 1} \}$ for some $k'$ and $\ell$ in which:
\[ \sum_{j = 1}^{k'} a_j + (n - 2k + 1)C + \ell \cdot C = \sum_{j' = k' + 1}^n a_j' + (n - 2k + 1 - \ell)C.\]
This implies that 
\[ 2\ell \cdot C = \sum_{j' = k' + 1}^n a_j' - \sum_{j = 1}^{k'} a_j \leq \sum_{j' = k' + 1}^n a_j' < C,\]
which is a contradiction.
\end{proof}

\begin{proof}[Proof of Theorem \ref{thm:hardness_pecs}.]
Assume an instance $\cI$ of the \BP\ problem is given. Based in this instance, we define an instance $\cI'$ of \DSPS. Let $a_{max}$ denote the maximum value among the integers in $A$. Define $C$ as $1 / \varepsilon \cdot \sum_{j = 1}^{2n} a_j$, where $\varepsilon$ is chosen such that $1 / \varepsilon$ is a large but constant integer. Note that $C > 1 / \varepsilon \cdot a_{max}$. Let $\cI'$ have $2n$ tasks, where each task $i$ has width and height $C + a_i$ for $i \in [2n]$. Our goal is to schedule the $2n$ tasks into a path of $W=n \cdot C + B$ edges while minimizing the peak. Based on the hardness of the \BP\ problem, we show that it is hard to distinguish between the case where an schedule with peak $2C(1 + \varepsilon)$ exists and the case where the minimum peak is larger than $3(C + 1)-\varepsilon$. \\

\noindent
\textbf{Completeness:} Assume that $\cI$ is a \textsf{Yes} instance, meaning that a partitioning $A = A_1 \dot\cup A_2$ exists that satisfies the cardinality and sum constraints. Define two shelves of squares, $S_i = \{ j | a_j \in A_i  \}$ for $i = 1, 2$, and schedule them starting at the leftmost edge. The width of each shelf is equal to $n \cdot C + B$ and the peak is at most $2 \cdot (C + a_{max}) < 2C(1 + \varepsilon)$.\\

\noindent
\textbf{Soundness:} Now, consider a \textsf{No} instance $\cI$. We claim that, for the corresponding \DSPS instance, no schedule with peak smaller than or equal to $3(C + 1)-\varepsilon$ exists. For the sake of contradiction, assume that it is the case. Since the size of each task is at least $C + 1$, it means that in the optimal solution for $\cI'$, no three tasks use the same edge. Also, the total width of the tasks is $2W = 2(nC + b)$, so no edge can have less than two tasks. This allows us to split the tasks $S$ into two sets $S_1$ and $S_2$. We start at the leftmost edge and pick one of the two tasks placed on this edge arbitrarily and put in $S_1$. Since every edge has exactly two tasks, immediately to the right of this task at least one another task must start. We put it in $S_1$ as well and proceed until we reach the rightmost edge, breaking ties arbitrarily along the way. We set $S_2 = S \backslash S_1$. It remains to show that each set has exactly $n$ tasks. Assume otherwise; let $S_1$ be composed of tasks $s_1,\ldots, s_{n + k}$, and $S_2$ be the tasks $s_{n + k + 1}, \ldots, s_{2n}$ for some $k$, $1 \leq k \leq n$. Since the tasks are placed one next to the other in each shelf, we have that $(n + k)C + \sum_{j = 1}^{n + k} a_j = (n - k)C + \sum_{j' = n + k + 1}^{2n} a_{j'}$. Therefore $\sum_{j' = n + k + 1}^{2n} a_{j'} \geq 2k \cdot C$, which is a contradiction for any value of $k > 0$ by our choice of $C$.

As a result, assuming that NP $\neq$ P, no polynomial-time algorithm can approximate the \DSPS problem within a factor of $\frac{3(C + 1)-\varepsilon}{2C(1 + \varepsilon)} = 3/2- \varepsilon'$, for some $\varepsilon' = O(\eps)$. \end{proof}

\section{Comparison between \DSP and \GSP}\label{sec:DSPvsGSP}

In this Section we provide instances where a gap between the optimal values they achieve interpreted as DSP and GSP instances can be observed. First we discuss the general case and then the case of \DSPS. It is worth noticing that, for the general case, an analogous proof can be derived from the results in~\cite{BDGS15}.

\begin{lemma}\label{lem:SPvsPM} There exists an instance of DSP with optimal peak $4$ such that the corresponding GSP instance has optimal peak $5$.
\end{lemma}

\begin{proof} Consider the following \DSP instance $\cI$, where $W=7$ and the set of tasks consists of the following eight elements (see Figure~\ref{fig:gap_instance} for a depiction):
	\begin{itemize}
		\item Two tasks of width $2$ and height $3$ (tasks $1$ and $2$ in the figure),
		\item Two tasks of width $4$ and height $1$ (tasks $3$ and $4$ in the figure),
		\item One task of width $3$ and height $1$ (task $5$ in the figure), 
		\item One task of width $1$ and height $1$ (task $6$ in the figure), and
		\item Two tasks of width $1$ and height $2$ (tasks $7$ and $8$ in the figure).
	\end{itemize}
	
	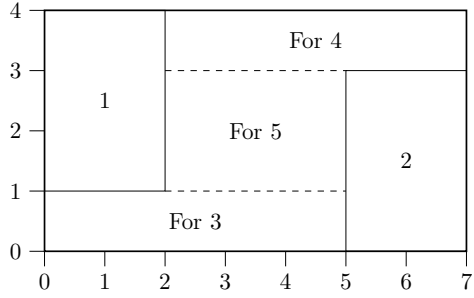
\begin{figure}
		\centering
		\resizebox{0.37\textwidth}{!}{\begin{tikzpicture}
		\draw[thick] (0,0) rectangle (7,4);
		\draw (0,0) -- (-0.25,0);
		\draw (-0.25,0) node[anchor=east] {$0$};
		\draw (0,1) -- (-0.25,1);
		\draw (-0.25,1) node[anchor=east] {$1$};
		\draw (0,2) -- (-0.25,2);
		\draw (-0.25,2) node[anchor=east] {$2$};
		\draw (0,3) -- (-0.25,3);
		\draw (-0.25,3) node[anchor=east] {$3$};
		\draw (0,4) -- (-0.25,4);
		\draw (-0.25,4) node[anchor=east] {$4$};
		\draw (-0.25,5) node[anchor=east] {\color{white!70}$5$};
		\draw (0,0) -- (0,-0.25);
		\draw (0,-0.25) node[anchor=north] {$0$};
		\draw (1,0) -- (1,-0.25);
		\draw (1,-0.25) node[anchor=north] {$1$};
		\draw (2,0) -- (2,-0.25);
		\draw (2,-0.25) node[anchor=north] {$2$};
		\draw (3,0) -- (3,-0.25);
		\draw (3,-0.25) node[anchor=north] {$3$};
		\draw (4,0) -- (4,-0.25);
		\draw (4,-0.25) node[anchor=north] {$4$};
		\draw (5,0) -- (5,-0.25);
		\draw (5,-0.25) node[anchor=north] {$5$};
		\draw (6,0) -- (6,-0.25);
		\draw (6,-0.25) node[anchor=north] {$6$};
		\draw (7,0) -- (7,-0.25);
		\draw (7,-0.25) node[anchor=north] {$7$};
		\draw (0,1) rectangle (2,4);
		\draw (1,2.5) node {$1$};
		\draw (5,0) rectangle (7,3);
		\draw (6,1.5) node {$2$};
		\draw[dashed] (2,1) -- (5,1);
		\draw[dashed] (2,3) -- (5,3);
		\draw (2.5,0.5) node {For $3$};
		\draw (3.5,2) node {For $5$};
		\draw (4.5,3.5) node {For $4$};
		\end{tikzpicture}}
		\caption{If we assume by contradiction that some optimal solution of height $4$ for the GSP instance described in Lemma~\ref{lem:SPvsPM} exists, it must have this structure.}\label{fig:gap_structure}
	\end{figure}
	
	As it is possible to see in Figure~\ref{fig:gap_instance}, the optimal solution has peak $4$ (since $OPT\ge a(\cI)/W=4$). We will show now that there is no solution for the corresponding GSP instance of height $4$, which would conclude the proof.
	
	Suppose by contradiction that there exists a solution to the corresponding GSP instance of height $4$. Let us imagine for the sake of presentation that we draw a grid of unit-size cells over the rectangular region $[0,7] \times [0,4]$, defining four rows of height $1$ and seven columns of width $1$. First of all, notice that in any feasible packing of the rectangles into the region, rectangles $1$ and $2$ cannot be touching the top (resp. bottom) boundary of the region at the same time. If that is the case, then the rectangles $3$ and $4$ do not fit in the region as they cannot be placed in the same row and none of them fits in the rows which are partially occupied by rectangles $1$ and $2$. So let us assume w.l.o.g. that rectangle $1$ touches the top boundary and rectangle $2$ touches the bottom boundary. Since they both partially occupy the middle rows of the region, rectangles $3$ and $4$ must be placed one touching the bottom boundary and the other touching the top boundary. This implies that rectangle $5$ has to be placed in one of the middle rows (in the other rows there is just one cell free), forcing us to place rectangles $1$ and $2$ one touching the left boundary and the other touching the right boundary (see Figure~\ref{fig:gap_structure}). Suppose rectangle $5$ is assigned to the second row from bottom to top (the other case being symmetric). Then in the two topmost rows we have to pack two rectangles of height $2$ plus a rectangle of width $4$ which is not possible as their total width is larger than the space left due to rectangle $1$. This contradicts the fact that there is a feasible solution for the \GSP instance $\cI$ of height $4$. \end{proof}

Now we will prove that even for the case of square tasks, the optimum packing for the two problems of \DSPS and \GSPS can exhibit a gap.


\begin{lemma}\label{lem:gap_squares} There exists an instance of \DSPS such that the optimal schedule has peak $11$ but every feasible solution for the corresponding \GSPS instance has height at least $12$. \end{lemma}
\begin{proof} 
Consider a \DSPS with $W=13$ and containing the following set $\cI$ of tasks (see Figure~\ref{fig:gap_instance_PEC_squares} for a depiction):
\begin{itemize}
    \item Two tasks of height/width $6$ (tasks $1$ and $2$ in the figure),
    \item Two tasks of height/width $5$ (tasks $3$ and $4$ in the figure),
    \item One task of height/width $3$ (task $5$ in the figure),
    \item Two tasks of height/width $2$ (tasks $6$ and $7$ in the figure), and
    \item Four tasks of height/width $1$ (tasks $8,9,10$ and $11$ in the figure).
\end{itemize}

Since $a(\cI)=11\cdot 13$, we have that $OPT\ge 11$. Figure~\ref{fig:gap_instance_PEC_squares} shows that the optimal peak is at most $11$ and hence it is exactly $11$.

    
Assume by contradiction that there exists a feasible packing for the corresponding \GSPS instance of height at most $11$. Consider $\mathcal{K}$ to be the region $[0,0]\times[13,11]$ in the plane, and let $(x_i,y_i)$ be the coordinate of the bottom-left corner of task $i$ in the solution. Notice that $\mathcal{K}$ must be completely filled with tasks. 

We can assume that $x_1\le x_2$, and since tasks $1$ and $2$ have height $6$ and the height of $\mathcal{K}$ is $11$, it must hold that $x_1\le x_2+6$. Hence, w.l.o.g. there are two cases to consider:\begin{itemize}
    \item $x_1=0$ and $x_2=6$:\\
    In this case the region $[12,y_2]\times[13,y_2+6]$ can only contain squares of size $1$, and they cannot fill the region completely, so this case cannot happen.
    \item $x_1=0$ and $x_2=7$:\\
    We show that $y_1,y_2\in \{0,5\}$; Assume that $y_1 \not \in \{0,5\}$. Then tasks $2,3$ and $4$ must be packed inside the region $[6,0]\times[13,11]$ since they can not be packed above or below task $1$. Since $a(j_2)+a(j_3)+a(j_4)>77$, this is not possible, hence proving the claim.
    
    Note that if $y_1=y_2$ then, similarly to the previous case, the area in $[6,y_1]\times[7,y_1+6]$ can only contain tasks of size $1$ and they cannot fill this region completely. So we can assume that $(x_1,y_1)=(0,0)$ and $(x_2,y_2)=(7,5)$.
    
    Now every remaining rectangle is either packed in $[6,0]\times[13,5]$ or in $[0,6]\times[7,11]$. However, among tasks $3, 4$ and $5$, it is not possible to place two of them in one of the previously mentioned rectangular region together, contradicting the existence of a feasible solution of height $11$.
\end{itemize}
\end{proof}

\section{A PTAS for \DSP with short tasks}\label{sec:light-PTAS}

In this section we will prove Theorem~\ref{lem:light-PTAS} restated below.

\lightPTAS*

Before proving the result in detail we provide a couple of required technical lemmas regarding the computation of $\pi$-left-pushing of a given schedule $P(\cdot)$. First of all, we prove that such a solution can be indeed computed efficiently.

\begin{lemma}\label{lem:leftPushingPolyTime}
    Given a feasible schedule $P(\cdot)$ with peak $\pi$ for an instance $\cI$, one can compute a $\pi'$-left-pushing of $P(\cdot)$, with $\pi'\ge\pi$, in polynomial time.
\end{lemma}
\begin{proof}
    Let $1,...,n$ be the tasks sorted according to their starting edge in $P(\cdot)$ from left to right. Let $S_i$ be the starting edge of task $i$. First, inductively, we compute a $\pi'$-left-pushing of $\cI\setminus\{n\}$ and do not left-shift task $n$. Since we only left-shifted the tasks, the demand on the edges from $S_n$ to $e_W$ cannot increase. Thus, we reach a feasible solution such that its peak does not exceed $\pi'$. Now we compute the starting time of task $n$, $s^*$, if we left-shift this task as much as possible. Note that $s^*$ can only be either the leftmost edge or some edge $e$ such that some previous task finishes next to the left of $e$, as otherwise at least one more unit of left-shifting is possible for task $n$. Now, using this fact, we have at most $n$ possibilities for $s^*$ and we can compute this value in polynomial time. Note that if we call the obtained schedule as $P'(\cdot)$, then $P'(\cdot)$ is indeed a $\pi'$-left-pushing of $P(\cdot)$.     
\end{proof}

The following lemma summarizes the useful properties we can get when computing a left-pushing.

\begin{lemma}\label{lem:prop-left-push} Given a feasible schedule $P(\cdot)$ with peak $\pi$ for an instance $\cI$, the $\pi'$-left-pushing of $P(\cdot)$ for $\pi'\ge \pi$, let us say $P'(\cdot)$, satisfies the following properties:\begin{enumerate}
    \item There exists a node $t^*$ such that $P'(\cdot)$ is $(\pi'-h_{\max}(\cI), t^*)$-sorted, and
    \item every $i\in \cI$ has a starting edge in $\mathcal{E}'$ of the form $\sum_{j\in \cI'}{w(j)}$ for some $\cI' \subseteq \cI\setminus \{i\}$ ($0$ if $\cI'$ is empty).
\end{enumerate}\end{lemma}

\begin{proof}

We now show a proof of the two properties:\\

\noindent
\textbf{1.} Suppose that there exists a node $k$ such that the demand on the edge to the left of $k$ is smaller than $\pi'-h_{\max}(\cI)$ and the demand on the edge to the right of $k$ is larger than $\pi'-h_{\max}(\cI)$. This implies that some task starts at the edge to the right of $k$, but then it is possible to left-shift this task without surpassing the threshold of $\pi'$ which is a contradiction. At this point we know that there exists $k'$ such that every edge to the left of $k'$ has demand larger than $\pi'-h_{\max}(\cI)$, and let $k^*$ be the rightmost such node. Similarly to the previous case, if after $k^*$ there exists a node $k$ such that the load to the left of $k$ is smaller than the demand to the right of $k$, then again there must exist a task starting to the right of $k$ and, since their demands are at most $\pi'-h_{\max}(\cI)$, left-shifting such task does not violate the threshold of $\pi'$ which is a contradiction.\\

\noindent
\textbf{2.} Suppose there exists a task not satisfying the claim, and let $i$ be the leftmost such task in $P'(\cdot)$. It is easy to see that $i$ cannot start at the leftmost edge and also that the demand on the edge just to the left of $P'(i)$ is larger than $\pi'-h(i)$ as otherwise a left-shifting of $i$ is possible. Due to $i$ being the leftmost task, no task $i'$ can finish just to the left of $P'(i')$, as otherwise the number of edges before $P'(i)$ would be the sum of some widths in $\cI$ plus $w(i')$, thus fulfilling the claim for $i$. This implies that every task using the edge just to the left of $P'(i)$ must also use edge $P'(i)$. But then the total demand just to the left of $P'(i)$ would be at most the total demand on $P'(i)$ minus $h(i)$, which is at most $\pi'-h(i)$.
\end{proof}

We can now proceed with the proof of Lemma~\ref{lem:light-PTAS}, where at some point in the proof we will make use of the following concentration bound which was proved in~\cite{CCKR11}.

\begin{lemma}\label{lem:calinescu}[Calinescu et al.~\cite{CCKR11}] Let $X_1, X_2, \dots, X_n$ be independent random variables and let $0\le\beta_1, \beta_2,\dots,\beta_n\le 1$ be real numbers, where for each $i=1,2,\dots,n$, $X_i=\beta_i$ with probability $p_i$ and $X_i=0$ otherwise. Let $X=\sum_{i=1}^{n}{X_i}$ and $\mu=\mathbb{E}[X]$. Then

\begin{enumerate}
    \item The variance of $X$, $\sigma^2(X)$, is at most $\mu$, and
    \item For any $0<\lambda<\sqrt{\mu}$, $\mathbb{P}[X>\mu+\lambda \sqrt{\mu}] < e^{-\frac{\lambda^2}{2}(1-\lambda/\sqrt{\mu})}$.
\end{enumerate}

\end{lemma}

\begin{proof}[Proof of Theorem~\ref{lem:light-PTAS}] 

Let $\delta>0$ be a constant that we will specify later. We will partition the tasks into two sets according to their widths: we will say that a task $i$ is \emph{horizontal} if $w(i)>\delta\cdot W$ and otherwise we will say it is \emph{narrow}. Consider by now only the horizontal tasks in $\cI$, and assume that the value $OPT$ is known. Thanks to Lemma~\ref{lem:prop-left-push}, by computing an $OPT$-left-pushing of the optimal solution, we know there exists a set $\mathcal{E}_H\subseteq E$ that can be computed in polynomial time such that the starting edge of every task belongs to $\mathcal{E}_H$. Indeed, edges in $\mathcal{E}_H$ correspond to the sum of widths of some horizontal tasks, implying that the number of widths in the sum must be at most $\frac{1}{\delta}$. Hence, all the possible starting edges are of the form $\displaystyle\sum_{i\in \cI'}{w(i)}$ where $|\cI'| \le \frac{1}{\delta}$. The set $\mathcal{E}_H$ consisting of these edges has size at most $n^{1/\varepsilon - 1}$ and can clearly be computed in polynomial time.

With the following integer program we can compute a feasible solution corresponding to a $OPT$-left-pushing of some scheduling for these tasks. We define a variable $x_{i,k}$ for each task $i$ and starting edge $k$ in the previously computed set $\mathcal{E}_H$ (if task $i$ cannot be scheduled starting at edge $k$ this variable is not considered): \[\begin{array}{rrcl} \min & \lambda & & \\ s.t. & \displaystyle\sum_{k\in \mathcal{E}_H}{x_{i,k}} & = & 1 \qquad \qquad \forall i \text{ horizontal} \\ & \displaystyle\sum_{i\text{ hor. }}{\displaystyle\sum_{k' \in \mathcal{E}_H(i,q)}{h(i) \cdot x_{i,k'}}} & \le & OPT  \qquad \forall q \in \mathcal{E}_H \\ & x_{i,k} & \in & \{0,1\} \qquad \forall i\text{ horizontal}, k \in \mathcal{E}_H, \end{array}\] where, given $i\in \cI$ and $q\in \{1,\dots,W\}$, $\mathcal{E}_H(i,e)$ is the set of edges $k\in \mathcal{E}_H$ such that, if $i$ has $k$ as starting edge, then it uses edge $e$. In other words, the second family of constraints is ensuring that the total demand of the constructed solution is at most $OPT$ in every edge (which can be done with polynomially many constraints thanks to the size of $\mathcal{E}_H$).

We will consider the canonical linear relaxation of the formulation, and let $\Vec{x}$ be an optimal solution to this LP (which can be computed in polynomial time). In order to derive a feasible solution we will use \emph{Randomized Rounding with Alterations}, a technique previously used in similar settings for Packing and Scheduling problems~\cite{CCKR11,MW15,AW15}. In a first stage, for each task $i$, we will sample one starting edge $k$ according to the probability distribution induced by $\{x_{i,k}\}_{k\in\mathcal{E}_H}$. Now, in a second stage, we scan the starting edges $k$ from left to right, and the sampled tasks $i$ starting at node $k$ according to the sample in any order, and we add $i$ to the current solution as long as the obtained peak is no more than $(1+\varepsilon)OPT$. Observe that this is a dependent rounding where each task $i$ is finally scheduled in the solution with marginal probability at most $x_{i,k}$.

Suppose we are applying the previous procedure, and let $k$ be a fixed edge in that order. Let $\tilde{X}_{i,k} \in \{0,1\}$ be equal to $1$ if and only if $i$ is scheduled starting at edge $k$ in the first stage, and similarly we define $\tilde{Y}_{i,k}$ to be $1$ if and only $i$ is scheduled starting at edge $k$ in the second stage. Notice that $\tilde{Y}_{i,k} \le \tilde{X}_{i,k}$ deterministically. By stochastic domination, we have that \[\mathbb{P}\left[ \displaystyle\sum_{i\text{ hor. }}{\tilde{Y}_{i,k}\cdot h(i)} > (1+\varepsilon)OPT\right] \le \mathbb{P}\left[ \displaystyle\sum_{i\text{ hor. }}{\tilde{X}_{i,k}\cdot h(i)} > (1+\varepsilon)OPT\right].\] To upper bound the latter quantity we will consider two cases:

\begin{itemize}
    \item If $\mu\le\frac{3}{4\delta}$, then we can use Chebyshev's inequality for the variable $Z:=\displaystyle\sum_{i\text{ hor. }}{\frac{\tilde{X}_{i,k}\cdot h(i)}{\delta OPT}}$ (notice that thanks to Lemma~\ref{lem:calinescu} it holds that $\sigma(Z) \le \sqrt{\mu}$), from where we obtain that \begin{eqnarray*} \mathbb{P}\left[ \displaystyle\sum_{i\text{ hor. }}{\tilde{X}_{i,k}\cdot h(i)} > (1+\varepsilon)OPT\right] & = & \mathbb{P}\left[Z>\frac{1+\varepsilon}{\delta}\right] \\ & \le & \mathbb{P}\left[|Z-\mu| > \left(\frac{1+\varepsilon}{\delta} - \frac{3}{4\delta}\right) \cdot \frac{\sigma(Z)}{\sqrt{\mu}}\right] \\ & \le & \frac{16\mu \delta^2}{(1+4\varepsilon)^2} \le \varepsilon
    \end{eqnarray*} for $\delta \le \frac{\varepsilon}{4}$.
    
    \item If $\mu>\frac{3}{4\delta}$, we first set $\lambda = \frac{1+\varepsilon - \mu\delta}{\delta \sqrt{\mu}}$ so that $\mu + \lambda\sqrt{\mu} = \frac{1+\varepsilon}{\delta}$. Notice that $\mu = \displaystyle\sum_{i\text{ hor. }}{\frac{x_{i,k}\cdot h(i)}{\delta OPT}} \le \frac{1}{\delta}$ due to the constraints in the LP.
    
    Now, it is not difficult to see that $\lambda$ is decreasing as a function of $\mu$, implying that $\lambda \ge \frac{1+4\varepsilon}{\sqrt{12\delta}}$. Furthermore, we have that $1-\frac{\lambda}{\sqrt{\mu}} = 2-\frac{1+\varepsilon}{\delta\mu} \ge \frac{2}{3}$, and thus also $\lambda<\sqrt{\mu}$. Now we can use Lemma~\ref{lem:calinescu} applied to the variables $\{X_{i,k} h(i)/(\delta OPT)\}_{i\text{ hor. }}$ and their sum $Z$ and obtain \begin{eqnarray*} \mathbb{P}\left[ \displaystyle\sum_{i \text{ hor. }}{\tilde{X}_{i,k}\cdot h(i)} > (1+\varepsilon)OPT\right] & = & \mathbb{P}\left[Z>\mu + \lambda \sqrt{\mu}\right] \\ & < & e^{-\frac{\lambda^2}{2}(1-\lambda/\sqrt{\mu})} \\ & < & e^{-\frac{2}{9} \frac{(1+4\varepsilon)^2}{12\delta}} \le \varepsilon
    \end{eqnarray*} for $\delta\le \frac{(1+4\varepsilon)^2}{54} \ln{\frac{1}{\varepsilon}}$.
\end{itemize}

This implies that we get a solution with peak at most $(1+\varepsilon)OPT$ and the probability that a task is not scheduled is at most $\varepsilon$. As a consequence, in expectation the total area of tasks that were not placed is at most $\varepsilon W \cdot OPT$, and hence using Markov's inequality we get that the probability that these tasks have area larger than $2\varepsilon W \cdot OPT$ is at most $\frac{1}{2}$. Thus, if the area of these tasks is at most $2\varepsilon W \cdot OPT$ and since their heights are at most $\delta \cdot OPT$, we can place them into a rectangular region of height $4\varepsilon OPT$ and width $W$ using Corollary~\ref{cor:2-apx}. If the area guarantee is not satisfied then we repeat the whole process to ensure it as, in expectation, a constant number of times only is required.

Now we will include the set $\mathcal{N}$ of narrow tasks into the solution by applying Lemma~\ref{lem:packing-narrow-Qtsorted} with parameter $\pi=(1+5\varepsilon)OPT$. Consider a $\pi$-left-pushing of the solution. Thanks to Lemma~\ref{lem:prop-left-push}, there exists a node $t^*$ such that the obtained schedule is $(\pi,t^*)$-sorted.
Furthermore, it is not difficult to see that $\pi\ge (1+4\varepsilon)\frac{a(\cI)}{W} + h_{\max}(\cI'')$ and $w_{\max}(\mathcal{N}) \le \varepsilon W \le \frac{4\eps}{2(4\eps+1)}$ for $\eps\le 1/4$, hence satisfying the requirements of the lemma. 
This way, we obtain a feasible scheduling with peak at most $(1 + 5\varepsilon)OPT$.

Finally, in order to avoid knowing the value of $OPT$, we can approximately guess it using any constant approximation (such as Corollary~\ref{cor:2-apx}) and define a constant number of candidates. \end{proof}






\section{\DSPS}
\label{sec:DSPS}

In this section we will discuss our algorithmic results for the special case of DSP restricted to square tasks, denoted as \DSPS. We will first provide a couple of useful known results plus a short discussion about \GSPS, and then we discuss the details of our main result.

A very useful technique to place rectangles into a region based almost only on their total area is Steinberg's algorithm~\cite{Steinberg97}, originally devised as an approximation algorithm for \GSP. The following theorem summarizes the required properties to obtain a packing, which we will use as subroutine in some of our results. Also from here it can be noticed that Steinberg's algorithm is a $2$-approximation for \DSP.

\begin{theorem}[Steinberg~\cite{Steinberg97}]\label{thm:Steinberg} Suppose we are given a rectangular region $B$ of height $h(B)$ and width $w(B)$ and a set of rectangles $\mathcal{R}$ such that $h_{\max}(\mathcal{R}) \le h(B)$, $w_{\max}(\mathcal{R}) \le w(B)$, and \[a(B) \ge 2a(\mathcal{R}) + (2h_{\max}(\mathcal{R}) - h(B))_+ (2w_{\max}(\mathcal{R}) - w(B))_+,\] where $(x)_+$ is $max(x,0)$ then it is possible to embed $\mathcal{R}$ into $B$ in polynomial time. \end{theorem}

The following two known results have been applied in the context of \GSP. We want to remark that the guarantees the results provide do not directly hold for the case of \DSP as they return non-overlapping embedding of rectangles in the plane.

\begin{theorem}[Bansal et al.~\cite{BCJPS09}]\label{thm:bansal} Given a rectangular region $B$ and a set of rectangles $\cI$ that can be embedded non-overlappingly into the region, it is possible to pack, for any $\varepsilon>0$, a set $\cI'\subseteq\cI$ into $B$ in polynomial time such that $a(\cI') \ge (1-\varepsilon)a(\cI)$. \end{theorem}

\begin{theorem}[G\'{a}lvez et al.~\cite{GGJJKR20}]\label{thm:skewed-SP} For any $\varepsilon>0$, there exists $\delta>0$ such that it is possible to compute in polynomial time a $(3/2+\varepsilon)$-approximate solution for any instance of \GSP satisfying that no rectangle has height larger than $\delta OPT$ and width larger than $\delta W$. \end{theorem}

These two known results allow to obtain an almost tight $(3/2+\varepsilon)$-approximation for \GSPS.

\begin{observation}\label{obs:3-2-SP} For any $\varepsilon>0$, there exists a $(3/2+\varepsilon)$-approximation for \GSPS. \end{observation}

\begin{proof} Let $\delta$ be the constant from Theorem~\ref{thm:skewed-SP}. We distinguish two cases:

\begin{itemize}
    \item if $W<\frac{1}{\delta} \cdot OPT$, then using Theorem~\ref{thm:bansal} with parameter $\varepsilon' = \varepsilon^2 \delta$ we can pack all the rectangles but a subset $A$ with total area of at most $(\varepsilon^2 \delta)\cdot W \cdot OPT$ in polynomial time into a rectangular region of width $W$ and height $OPT$.
    Since $({\varepsilon^2\delta})\cdot W \cdot OPT<\varepsilon^2 OPT^2$, any square in $A$ has height and width at most $\varepsilon\cdot OPT$. Now one can pack all the rectangles in $A$ into a rectangular region of width $W$ and height $2\varepsilon\cdot OPT$ using Steinberg's algorithm (Theorem~\ref{thm:Steinberg}). 
    \item If $W\ge \frac{1}{\delta} \cdot OPT$, since every square has width at most $\delta W$, we can use Theorem~\ref{thm:skewed-SP} to get a $(3/2+\varepsilon)$-approximation in this case.
\end{itemize}
\end{proof}

\subsection{Proof of Theorem~\ref{thm:squareDSP}}


We now present a $3/2$-approximation for \DSPS, that actually works in the more general framework of \emph{bounded aspect ratio instances}: for \afr{a constant} $\beta\ge 1$, an instance of \DSP has aspect ratio at most $\beta$ if, for each task $i$, it holds that $h(i) \le w(i) \le \alpha h(i)$\footnote{In the literature it is usually defined an instance as having aspect ratio at most $\alpha$ if, for each task $i$, it holds that $\frac{1}{\alpha} \le \frac{h_i}{w_i} \le \alpha$. It is not difficult to verify that the two definitions are equivalent by appropriately scaling the widths.} (in particular square tasks have aspect ratio $1$). 

\begin{theorem}[Restatement of Theorem~\ref{thm:squareDSP}]\label{thm:3-2-PEC-squares} Given $\beta>0$ constant, there exists a $3/2$-approximation for \DSP with aspect ratio at most $\beta$. \end{theorem}

We will make use of the following result which is similar in spirit to Theorem~\ref{thm:bansal} but for the case of \DSP, which we will prove later (see Section~\ref{sec:Bansal-DSP}).

\begin{lemma}\label{lem:Bansal-PEC} Given $\varepsilon>0$ and an instance $\cI$ of \DSP with optimal peak $OPT$, it is possible to partition $\cI$ into two sets $\cI', \cI''$ such that \begin{itemize}
    \item $a(\cI'')\le \varepsilon\cdot W \cdot OPT$, and
    \item It is possible to compute in polynomial time a schedule of peak $(1+\varepsilon)OPT$ for $\cI'$.
\end{itemize} \end{lemma}


\begin{proof}[Proof of Theorem~\ref{thm:3-2-PEC-squares}] Let $\cI=\{1, 2,\dots, n\}$ be an instance of \DSPS, where the tasks are sorted non-increasingly by height, having optimal peak $OPT$. We will distinguish two cases depending on the relation between $W$ and $OPT$.

If $W\le 100 \beta \cdot OPT$, then we can apply Lemma~\ref{lem:Bansal-PEC} with parameter $\varepsilon^2/100$, $\varepsilon\le \frac{1}{4\beta}$, hence obtaining a feasible schedule of peak at most $(1+\varepsilon^2/100) OPT$ for almost the whole instance except for a subset of tasks of total area at most $\frac{\varepsilon^2}{100} W \cdot OPT$. Since $W\le 100\beta \cdot OPT$, no such task can have height larger than $\varepsilon\sqrt{\beta} OPT$, and thus we can place all these tasks into an extra rectangular region of width $W$ and height $2\varepsilon \sqrt{\beta}OPT$ using Steinberg's algorithm (Theorem~\ref{thm:Steinberg}), and place this region on top of the current solution to obtain a schedule of peak at most $(1+3\varepsilon\sqrt{\beta})OPT$.

Consider now the case when $W>100\cdot\beta OPT$. Let $i_1$ be the largest index such that $h(i_1)>0.49\cdot OPT$, and let $i_2$ be the largest index such that $\sum_{i=1}^{i_2}w(i)<W$. Note that for every item $i$ such that $i>i_2$, $h(i)\le OPT/2$. We consider the following two possibilities:\\

\noindent
\textbf{Case 1:} $\sum_{i=1}^{i_1}w(i)>1.8\cdot W$. Then clearly $i_2<i_1-2$ as $w(i) \le W/100$. We will place task $1$ starting on the leftmost edge and, for every $1<i\le i_2$, we will place task $i$ side by side to the right of task $i-1$. We then place task $i_2+2$ finishing at the rightmost edge, and then for every $i_2+2<i<i_1$ we place task $i$ side by side to the left of task $i-1$ (see Figure~\ref{fig:PEC_squares_3by2}).

\begin{figure}
    \centering
    \resizebox{0.5\textwidth}{!}{\begin{tikzpicture}[scale=0.6]
    \draw[thick] (0,0) -- (20,0);
    \draw[color=gray!70] (0,0) -- (0,11.5);
    \draw[color=gray!70] (20,0) -- (20,11.5);
    
    \draw[dashed] (0,0) -- (-0.25,0);
    \draw (-0.25,0) node[anchor=east] {$0$};

    \draw[dashed] (20,8) -- (-0.25,8);
    \draw (-0.25,8) node[anchor=east] {$OPT$};

    \draw[dashed] (20,11) -- (-0.25,11);
    \draw (-0.25,11) node[anchor=east] {$\frac{3}{2}\cdot OPT$};
    \draw (20.25,11) node[anchor=west] {\textcolor{white}{$\frac{3}{2}\cdot OPT$}};

    
    \draw (0,0) -- (0,-0.25);
    \draw (0,-0.25) node[anchor=north] {$0$};
    
    \draw (20,0) -- (20,-0.25);
    \draw (20,-0.25) node[anchor=north] {$W$};
    \draw (0,0) rectangle (4.5,4.5);
    \draw (2.25,2.25) node {$1$};
    \draw (4.5,0) rectangle (8.8,4.3);
    \draw (6.65,2.15) node {$2$};
    
    \draw (12.4,2.125) node {$\dots$};
    
    \draw (16,0) rectangle (20,4);
    \draw (18,2) node {$i_2$};
    
   
   \draw (16.1,4.1)   rectangle    (20,8);
    \draw (18.05,6.05) node {$j_{i_2+2}$};
    \draw (12.2,4.2)   rectangle    (16.1,8);
    \draw (14.1,6.1) node {$j_{i_2+3}$};

    \draw (10.7,6.15) node {$\dots$};

    \draw (5.5,4.4)   rectangle    (9.1,8);
    \draw (7.3,6.2) node {$j_{i_1-1}$};
    
    \draw (10,9.5) node {Steinberg's algorithm};
    
    \end{tikzpicture}}
    \caption{A $\frac{3}{2}$-approximation for \DSPS, depiction of the case when $\sum_{i=1}^{i_1}w(i)>1.8\cdot W$ from the proof of Theorem~\ref{thm:3-2-PEC-squares}.}
    \label{fig:PEC_squares_3by2}
\end{figure}
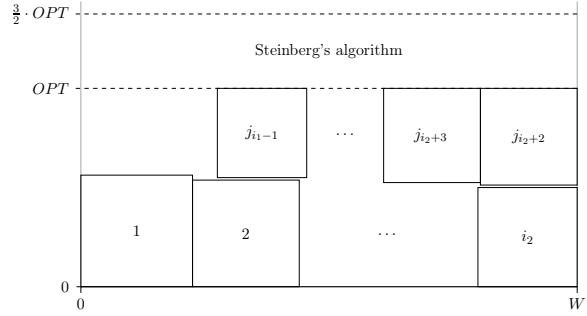

Note that the tasks which we already scheduled have total width at least $1.8W-w(i_2+1)-w(i_1)>1.8W-0.02W=1.78W$. Thus, their total area is at least $1.78W\cdot 0.49 OPT>0.8 W\cdot OPT$.
We show that this schedule has peak at most $OPT$. First of all, note that $\sum_{i=1}^{i_1}w(i) \le 2\cdot W$ since, in the optimal solution, every edge can be used by at most two such tasks. 
Assume that there exist two tasks $p, q$, $p<i_2+1<q<i_1$ such that their subpaths in the schedule overlap and $h(p)+h(q) > OPT$. Now consider the optimal solution and let $J_1=\{1,\dots,p\}$ and let $J_2=\{{p+1},\dots,q\}$. Since the tasks are sorted non-increasingly by height, subpaths of tasks from $J_1$ in the schedule cannot overlap. Also, all the pairs $i,{i'}$ such that $i \in J_1$ and ${i'}\in J_2$ cannot overlap in the schedule because of the same reason. Since no three tasks from $J_2$ can overlap in the schedule and $\sum_{i=p+1}^{q}w(i)>2(W-\sum_{j=1}^{p}w(j))$, this is a contradiction.
    
Now the rest of the tasks have total area at most $0.2 OPT\cdot W$. Since these tasks have width at most $0.5\beta OPT \le 0.005 W$, we can use Steinberg's algorithm (Theorem~\ref{thm:Steinberg}) to place them into an extra rectangular box of height $\frac{1}{2}OPT$ and width $W$, and place this box on top of the current solution. \\

\noindent
\textbf{Case 2:} $\sum_{i=1}^{i_1}w(i) \leq 1.8W$. We schedule tasks $1,\dots,i_2$ in the same way as we did for the previous case. Now schedule task $i_2+1$ starting on the leftmost edge, and for $i_2+1< i\le i_1$ we schedule task $i$ side by side to the right of task ${i-1}$. 
Since for every $i>i_2$ we have that $h(i)\le 0.5OPT$, then the current schedule does not exceed peak $1.5OPT$ and it is sorted (as it is the sum of two non-increasing demand profiles). Notice that this schedule is well defined as $W-0.01W<\sum_{i=1}^{i_2}w(i) \le W$ and $\sum_{i=i_2+1}^{i_1}{w(i)} \le 1.8W - (0.99W) \le 0.81\cdot W$. We denote the set of remaining tasks by $\cI'$.

Let us first approximately estimate $OPT$ from below, meaning that we compute a value $L$ such that $(1-\varepsilon)OPT \le L \le OPT$ for some given $\varepsilon>0$. This is possible to do by using any constant approximation for \DSPS (for example Corollary~\ref{cor:2-apx}) and then approximately guessing the value. Now we will place the remaining tasks by means of Lemma~\ref{lem:packing-narrow-sorted} with parameter $\pi=\max\{h(1)+h(i_2+1), 3L/2\}$. It is not difficult to verify that $\pi\ge OPT + h_{\max}(\cI')$, $w_{\max}(\cI') \le 0.49\cdot\beta OPT \le 0.0049W$, $h_{\max}(\cI') \le 0.49\cdot OPT$ and $(1-0.0049)(1.01-1.5\cdot\varepsilon)W\cdot OPT \ge W\cdot OPT$ for $\varepsilon$ small enough. Thus, we obtain a feasible schedule of peak at most $\frac{3}{2}OPT$. \end{proof}

\subsubsection{Proof of Lemma~\ref{lem:Bansal-PEC}}\label{sec:Bansal-DSP}

Given $0<\mu<\delta$, we will start by classifying the tasks as follows:

\begin{itemize}
    \item A task is \emph{big} if $h(i)>\delta OPT$ and $w(i)>\delta W$;
    \item A task is \emph{wide} if $h(i)\le \mu OPT$ and $w(i)>\delta W$;
    \item A task is \emph{long} if $h(i)>\delta OPT$ and $w(i)\le \mu W$;
    \item A task is \emph{tiny} if $h(i)\le \mu OPT$ and $w(i)\le \mu W$; and
    \item A task is \emph{intermediate} if $\mu OPT < h(i) \le \delta OPT$ and $\mu W < w(i) \le \delta W$.
\end{itemize}

Analogously to Lemma~\ref{lem:mediumrectanglesarea}, we can show that it is possible to choose $\mu$ and $\delta$ such that they differ by a large factor and that the total area of intermediate tasks is at most $\varepsilon^2\cdot OPT\cdot W$. From now on we will assume that $\mu$ and $\delta$ chosen like that, and we will discard the intermediate tasks from the instance (meaning that we include them into $\cI''$). Let us temporarily remove the tiny tasks, we will add them in the end via a slight modification of Lemma~\ref{lem:packing-narrow-sorted}.

Consider the optimal solution restricted to wide tasks and its corresponding demand profile $D$. We will prove in the following lemma that, by increasing the peak of the solution by $2\varepsilon OPT$, we can ``round-up'' the demand profile $D$ so that it has $O_{\eps}(1)$ jumps only.

\begin{lemma}\label{lem:discretize_curve} There exists a demand profile $D'$ that has $O_{\eps}(1)$ jumps, it upper bounds (vectorially) the demand profile $D$ induced by wide tasks, and satisfies that, for each edge, the difference between $D'$ and $D$ is at most $2\eps OPT$. \end{lemma}

\begin{proof}
Consider the demand profile $D$ induced by wide tasks in the optimal solution. Let us first define an auxiliary demand profile $D''$ corresponding to $D$ plus $2\eps OPT$ on each coordinate.

We will build our new demand profile $D'$ as follows: We start first at the leftmost edge $e_1$, storing its total demand $\ell_1$ plus $\eps OPT$. Let $e_2$ be the leftmost edge with demand either larger than $\ell_1+\varepsilon OPT$ or smaller than $\ell_1-\varepsilon OPT$. Then we store in $D'$ the edge $e_2$ and demand $\ell_2 + \eps OPT$, where $\ell_2$ is the total demand in edge $e_2$, and restart the process until the end of the path. We will also include in the demand profile all the edges which are multiples of $\delta W$, meaning that at each such edge $x$ we store it and also store its demand plus $\eps OPT$, and restart the process with that value. It is not difficult to see that this demand profile $D'$ is completely contained (vectorially) between $D$ and $D''$.

We will now argue about the number of jumps of $D'$. Let us partition the whole path into intervals of $\delta\cdot W$ contiguous edges starting at the leftmost edge, and consider any such interval. Notice that no wide task can start and end inside the same interval. If inside the interval there is a jump up in $D'$, this means that the demand has increased by at least $\varepsilon OPT$ from the previous stored edge, and this is due to a set of tasks of total height at least $\varepsilon OPT$ starting inside the interval. However, as mentioned before, the tasks starting in the interval must finish at a different interval. Similarly, if there is a jump down in $D'$ it is due to a set of tasks of total height at least $\varepsilon OPT$ finishing inside the interval, which again must start at different intervals. Since the process is restarted at every final edge of the intervals, all these tasks are different and hence there can be at most $2/\varepsilon$ such jumps inside the interval (as they all contribute to the demand at either $k\delta W$ or $(k+1)\delta W$). Including now the jumps at edges which are multiples of $\delta W$, the number of jumps in $D'$ is at most $\frac{2}{\varepsilon\delta} + \frac{1}{\delta}$.
\end{proof}

Consider now a schedule of peak $(1+2\varepsilon)OPT$ where all the wide tasks are scheduled below the demand profile $D'$ from Lemma~\ref{lem:discretize_curve}. Since the number of large tasks is constant, we can schedule them on top of $D'$ and obtain a new demand profile that still has $O_{\eps}(1)$ many jumps. Doing an analogous procedure to the one described in Lemma~\ref{lem:container-horizontal} it is possible to define a constant number of containers for wide and large tasks below the demand profile. All the tasks that were not placed have negligible area and hence we can just discard them. Furthermore, If we allow to vertically slice long tasks, then the difference between $(1+2\eps)OPT$ and $D'$ induces a packing of the sliced long tasks into $O_{\eps}(1)$ rectangular regions or \emph{boxes}. The following lemma allows us to turn this packing into a feasible container packing for almost all the long tasks.

\begin{lemma}\label{lem:mainLinearGrouping}
Consider the previous scheduling of (sliced) long tasks decomposed into a set $\cB$ of $K=O_{\varepsilon'}(1)$ rectangular boxes. Then, there exists a partition of the long tasks into two sets $\cV^{cont}$ and $\cV^{disc}$ such that:
\begin{enumerate}\itemsep0pt
\item $\cV^{cont}$ can be packed into a set of at most $K'=O_{\varepsilon'}(1)$ vertical containers, where each container is fully contained in some box in $\cB$.
\item $\cV^{disc}$ has total area at most $\varepsilon^2 \cdot W \cdot OPT$.
\item The sizes of the containers belongs to a set that can be computed in polynomial time.  
\end{enumerate}
\end{lemma}

\begin{proof} The first step in our construction is to round up the heights of the long slices to multiples of $\delta^2 OPT$. Since long tasks have height at least $\delta \cdot OPT$, this rounding increases the peak of the solution by at most $\varepsilon \cdot OPT$ (we increase the height of the boxes accordingly). Observe that the number of possible distinct heights is at most $1/\delta^2$. 

Let us focus on a specific box $B\in \cB$ of size $w(B)\times h(B)$, and let $\cV_{sliced}(B)$ be the slices contained in $B$. Next, we partition $B$ into unit width stripes, and we shift slices in each stripe as down as possible. We call a \emph{configuration} $C$ of a stripe the sequence of (enlarged) heights $(h_1,\ldots,h_q)$ sorted non-increasingly. Notice that each stripe can contain at most $(1+\varepsilon)/\delta$ slices, and hence the number of possible configurations is at most $(1/\delta^2)^{(1+\varepsilon)/\delta}$. 

We reorder the stripes in $B$ so that equal configurations appear consecutively. Suppose that the number of stripes in $B$ with a given configuration $C=(h_1,\ldots,h_q)$ is $w(C)$, and $A(C)$ is the corresponding region. We cover $A(C)$ by creating $q$ consecutive vertical containers of width $w(C)$ and height $h_1,\ldots,h_q$ respectively. The height of each container belongs to a set that can be computed in polynomial time (it is a multiple of $\delta^2 OPT$). In order to enforce the same property for their widths, we round down the width of each such container to the largest multiple $w'(C)$ of $\frac{\mu}{\varepsilon} W$ not larger than $w(C)$. The number of these containers is $n_{cont}\leq K(1/\delta^2)^{(1+\varepsilon)/\delta}$.

We next use the obtained containers to place most of the tasks. We consider the containers in non-increasing order of height and the slices of long tasks in the same order, breaking ties so that slices of the same task appear consecutively. We also create a dummy final container of sufficient width and of height large enough to accommodate the total width of the slices minus the total width of the containers. 
Now, we place back the slices into the containers following the previous order. Notice that all the slices will fit. We discard each task whose slices are contained in two containers (three is not possible) and all the tasks whose slices are contained in the dummy final container. The total area of these discarded tasks is at most \[n_{cont}\mu W\cdot OPT+n_{cont}\frac{\mu}{\varepsilon} W \cdot OPT.\] The above quantity is at most $\varepsilon^2 W \cdot OPT$ provided that $\mu$ is small enough.

All the tasks that are not placed, which have total area at most $\varepsilon^2 W \cdot OPT$, we include them into set $\cV^{disc}$, and the remaining tasks which are placed into the containers are included into set $\cV^{cont}$, satisfying the claims of the lemma. \end{proof}

This way we obtain a guessable container packing for all the non-tiny tasks of peak at most $(1+2\varepsilon)OPT$ except for a set of tasks of total area at most $\varepsilon^2 W\cdot OPT$, and this packing can be computed almost optimally by means of Lemma~\ref{lem:containersPackPTAS}. We will include the tiny tasks by applying a greedy procedure on top of the constructed solution as follows.

Consider an optimal schedule of the containers, and notice that its demand profile has $O_{\delta}(1)$ jumps. We will create rectangular regions according to the difference between $(1+2\varepsilon)OPT$ and the demand profile, which indeed induces $O_{\delta}(1)$ rectangular boxes starting and ending at the jumps of the demand profile. Let $K\in O_{\delta}(1)$ be the total number of boxes and containers. We will sort them arbitrarily and place the tiny tasks using a variation of the algorithm from Lemma~\ref{lem:packing-narrow-sorted}: We consider the tiny tasks in any fixed arbitrary order and start with the first box or container. We scan the whole list and place the current tiny task in the current edge if possible, otherwise we continue with the list. If we finished with the list we move to the next possible edge to the right and start again. If some task could not be packed in the box or container we move to the next one and start again.

Notice that initially the demand profiles inside each box or container are non-increasing, and hence we maintain the following invariant: in the current box $B$, to the left of the current edge, every edge has demand larger than $h(B)-\mu OPT$ and to the right the demand profile is non-increasing. Hence, if we reached the edge at position $w(B)-\mu W$ in the box or container and some task still has to be placed, the total demand inside $B$ is at least $(h(B)-\mu OPT)(w(B)-\mu W)$. 

If this procedure schedules all the tasks, we are done. So assume that some tiny task $i$ could not be packed. This means that in every box or container $B$ the demand of tasks is larger than $(w(B)-\mu W)(h(B)-\mu OPT)$, and hence the total demand of the current solution is larger than $(1+2\varepsilon)W\cdot OPT - 2\mu K  W \cdot OPT \ge W\cdot OPT$ provided that $\mu \le \varepsilon / K$ (this can be ensured by defining $f(x):= x/K$ in Lemma~\ref{lem:mediumrectanglesarea}), which would be a contradiction. This concludes the proof of Lemma~\ref{lem:Bansal-PEC}.

\end{document}